\documentclass[a4paper,UKenglish,cleveref, autoref, thm-restate]{lipics-v2021}
\pdfoutput=1 
\hideLIPIcs  
\usepackage{svg}
\usepackage[utf8]{inputenc}  
\usepackage{babel} 
\usepackage{graphicx}    
\usepackage{subcaption}
\usepackage{amsmath}         
\usepackage{amsthm}          
\usepackage{amssymb}
\usepackage{mathtools}
\usepackage[ruled]{algorithm}
\usepackage{algpseudocode}
\usepackage{braket}
\usepackage{comment}
\usepackage{enumerate}

\usepackage[textsize=tiny,textwidth=2cm,color=green!50!gray]{todonotes} 

\newcommand{\cont}[1]{\operatorname{CONT}\left(#1\right)}           
\newcommand{\prof}[1]{p\left(#1\right)}         
\newcommand{\vol}[2][d]{\operatorname{VOL}_{#1}\left(#2\right)}     
\newcommand{\fvol}[2][d]{\operatorname{FVOL}_{#1}\left(#2\right)}     
\newcommand{\surf}[2][d]{\operatorname{SURF}_{#1}\left(#2\right)}   
\newcommand{\opt}[2]{\operatorname{OPT}_{\mathrm{#1}}\left(#2\right)} 
\newcommand{\opts}[1]{\opt{strip}{#1}}
\newcommand{\optk}[1]{\opt{knapsack}{#1}}
\newcommand{\ext}{\operatorname{ext}}
\newcommand{\sig}[2]{\sigma_{#1}\left(#2\right)}
\newcommand{\dsqks}[1]{{\sc {#1}-D Hc-Knapsack}} 
\newcommand{\dsqbp}[1]{{\sc {#1}-D Hc-Binpacking}} 
\newcommand{\dgks}[1]{{\sc {#1}-D Gen-Knapsack}} 
\newcommand{\dsqsp}[1]{{\sc {#1}-D Hc-StripPacking}} 

\newcommand{\sitems}{\mathcal{S}} 
\newcommand{\mitems}{\mathcal{M}} 
\newcommand{\litems}{\mathcal{L}} 
\newcommand{\Cname}[1]{C_{\mathrm{#1}}}
\newcommand{\Nname}[1]{N_\mathrm{#1}}
\newcommand{\up}{\mathrm{up}}
\newcommand{\down}{\mathrm{down}}

\newcommand{\abs}[1]{\left|#1\right|}               
\newcommand{\ceil}[1]{\left\lceil#1\right\rceil}    
\newcommand{\floor}[1]{\left\lfloor#1\right\rfloor} 
\newcommand{\intv}[1]{\left[#1\right]}              
\newcommand{\paren}[1]{\left(#1\right)}             
\newcommand{\nat}{\mathbb{N}}                       
\newcommand{\real}{\mathbb{R}}                      
\newcommand{\Vbox}{$\mathcal{V}$-Box}
\newcommand{\Nbox}{$\mathcal{N}$-Box}
\newcommand{\Vboxes}{$\mathcal{V}$-Boxes}
\newcommand{\Nboxes}{$\mathcal{N}$-Boxes}
\newcommand{\narrow}[0]{good}
\newcommand{\wide}[0]{bad}
\newcommand{\Th}{^{\mathrm{th}}}
\newcommand{\eps}{\varepsilon}
\renewcommand{\epsilon}{\varepsilon}                

\newcommand{\kvnr}[1]{}
\newcommand{\marr}[1]{}
\newcommand{\arir}[1]{}

\newcommand{\kvn}[1]{{#1}}
\newcommand{\mar}[1]{{#1}}
\newcommand{\ari}[1]{{#1}}

\title{A PTAS for Packing Hypercubes into a Knapsack} 

\author{Klaus {Jansen}}{Kiel University, Kiel, Germany}{kj@informatik.uni-kiel.de}{}{Research supported by German Research Foundation (DFG), project JA 612/25-1.}
\author{Arindam {Khan}}{Department of Computer Science and Automation, Indian Institute of Science, Bengaluru, India}{arindamkhan@iisc.ac.in}{0000-0001-7505-1687}{Research partly supported by Pratiksha Trust Young Investigator Award, Google India Research Award, and Google ExploreCS Award.}
\author{Marvin {Lira}}{Kiel University, Kiel, Germany}{stu204529@mail.uni-kiel.de}{}{Research supported by German Research Foundation (DFG), project JA 612/25-1.}
\author{K. V. N. {Sreenivas}}{Department of Computer Science and Automation, Indian Institute of Science, Bengaluru, India}{venkatanaga@iisc.ac.in}{}{}

\authorrunning{K. Jansen, A. Khan, M. Lira, and K. V. N. Sreenivas} 

\Copyright{Klaus Jansen, Arindam Khan, Marvin Lira, and K. V. N. Sreenivas} 



\category{} 

\relatedversion{} 

\acknowledgements{We thank Roberto Solis-Oba and three anonymous reviewers for their comments and suggestions on this paper.}

\nolinenumbers 

\EventEditors{Miko{\l}aj Boja\'{n}czyk, Emanuela Merelli, and David P. Woodruff}
\EventNoEds{3}
\EventLongTitle{49th International Colloquium on Automata, Languages, and Programming (ICALP 2022)}
\EventShortTitle{ICALP 2022}
\EventAcronym{ICALP}
\EventYear{2022}
\EventDate{July 4--8, 2022}
\EventLocation{Paris, France}
\EventLogo{}
\SeriesVolume{229}
\ArticleNo{76}
\begin{document}
\maketitle
\begin{abstract}
  
We study the $d$-dimensional hypercube knapsack problem (\dsqks{$d$}) where we are given a set of $d$-dimensional hypercubes  with associated profits, and
a knapsack which is a unit $d$-dimensional hypercube.
The goal is to find an axis-aligned non-overlapping packing of a subset of hypercubes such that the profit of the packed hypercubes is maximized.
For this problem, Harren (ICALP'06) gave an algorithm with an approximation ratio of  $(1+1/2^d+\eps)$. 
For $d=2$, Jansen and Solis-Oba (IPCO'08) showed that the problem admits a polynomial-time approximation scheme (PTAS);
Heydrich and Wiese (SODA'17) further improved the running time and gave an efficient polynomial-time approximation scheme (EPTAS).
Both the results use structural properties of 2-D packing, which do not generalize to higher dimensions.
For $d>2$, it remains open to obtain a PTAS, and in fact, there has been no improvement since Harren's result.

We settle the problem by providing a PTAS. 
Our main technical contribution is a structural lemma which shows that any packing of hypercubes
can be converted into another {\em structured} packing such that a high profitable subset of hypercubes is packed into a constant number of {\em special} hypercuboids,
called  \Vboxes~and \Nboxes. As a side result, we give an almost optimal algorithm for a variant of
the strip packing problem in higher dimensions. This might have applications for other multidimensional geometric packing problems.
\end{abstract}

\section{Introduction}\label{sec:introduction}

Multidimensional geometric packing problems are well-studied natural generalizations of the classical knapsack and bin packing problems.
In the $d$-dimensional geometric knapsack problem (\dgks{$d$}),
\kvn{where $d$ is a fixed constant parameter},
we are given a set of $n$ items $\mathcal{I}:=\{1, 2, \dots, n\}$.
Each item $i \in [n]$ is a $d$-dimensional ($d$-D) hypercuboid with side length $s_k(i) \in (0,1]$ along the $k\Th$ dimension and profit $p(i) \in \mathbb{Q}_{>0}$.
The goal is to pack a subset of hypercuboids into a $d$-D unit knapsack (i.e., $[0,1]^d$) such that the profit is maximized. 
We need the packing of the hypercuboids to be axis-aligned and non-overlapping.
In this paper, we study $d$-dimensional hypercube knapsack (\dsqks{$d$}), a special case of \dgks{$d$}, where all the items
are hypercubes, i.e., for all $i \in [n]$, the $i\Th$ hypercuboid is a hypercube of side length $s(i)$.

\dgks{$d$} generalizes the classical (1-D) knapsack problem \cite{assignment} and thus is NP-hard. 
It finds numerous applications in scheduling, ad placement, and cutting stock \cite{GalvezGHI0W17}.
For \dgks{$2$}, Jansen and Zhang \cite{jansen-zhang} gave a $(2+\eps)$-approximation algorithm.
G{\'{a}}lvez et al.~\cite{GalvezGHI0W17} gave a 1.89-approximation and $(3/2+\eps)$-approximation for the cardinality version (i.e., all items have the same profit).
Adamaszek and Wiese \cite{adamaszek2015knapsack} gave a QPTAS when the input size is quasi-polynomially bounded.  G{\'{a}}lvez et al.~\cite{Galvez00RW21} later gave a pseudo-polynomial time  $(4/3+\eps)$-approximation.
For \dgks{$3$}, the present best approximation ratio is $7+\eps$ \cite{diedrich2008approximation}. 
For $d \ge 4$, Sharma \cite{Sharma21fsttcs} has given a $(1+\eps)3^d$-approximation algorithm.
\ari{Interestingly for $d\ge2$, unlike \dgks{$d$}, \dsqks{$d$} is not a generalization of 1-D knapsack.
Leung et al. \cite{leung1990packing} showed \dsqks{$2$} is strongly NP-hard, using a reduction from the 3-partition problem.}
The NP-hardness status of \dsqks{$d$} for $d>2$, was open for a long time.
Recently, Lu et al.~\cite{lu2015packing, lu2015packing1} settled the status by 
showing that \dsqks{$d$} is also NP-hard for $d>2$.

A related problem is \dsqbp{$d$} where we are given  $d$-D hypercubes 
and the goal is to pack them into the minimum number of unit hypercubes. Back in 2006, Bansal et al.~\cite{bansal2006bin} gave an APTAS for this problem. 
\ari{Their algorithm starts by classifying the input set into small, medium, and large items.
They first pack the large items ($O(1)$ in number) near-optimally by brute force and then pack the small items using Next Fit Decreasing Height (NFDH)   \cite{coffman1980performance} in the gaps left in between the large items. 
The remaining unpacked (small and medium) items are packed into additional bins using NFDH. 
However, this \kvn{constructive approach} can not be used to devise a PTAS for \dsqks{$d$}.
Specifically, after packing the large items, the total space left to pack
the small items can be very small and this space \kvn{can be} fragmented among several voids in between the large items.}
This renders packing the small items difficult, especially if the small items occupy a small volume in the optimal solution and carry a lot of profit. This turns out to be a bottleneck case. Otherwise, if the volume of small items
in an optimal solution is significant, or if their profit is not significant, or if
the empty space in the optimal solution is significant, we can easily devise a PTAS. In fact, the  algorithm in \cite{bansal2006bin} can be adapted to devise PTASes for the special cases of \dsqks{$d$}: (i) cardinality case, and (ii) when each item has the profit:volume ratio in the range $[1,r]$ for a fixed constant $r$.
However, the case of arbitrary profits remains very difficult to handle.

For \dsqks{$d$}, Harren \cite{harren-journal} gave a $(1+{1}/{2^d}+\eps)$-approximation algorithm,  
by removing a least profitable large item and using the empty space to pack the small items. Interestingly, the approximation ratio gets smaller as $d$ grows. He also studied \dsqsp{$d$}, where, given a set of hypercubes and a strip with a $(d-1)$-dimensional base and unbounded height, the goal is to pack all the hypercubes while minimizing the height. He gave an APTAS for the special case of \dsqsp{$d$} when the ratio between the shortest and longest sides of the base is \kvn{bounded by} a constant.

Jansen and Solis-Oba \cite{jansen-journal} gave a PTAS for \dsqks{$2$}, by overcoming the above-mentioned bottleneck.
\kvn{
	They consider an optimal packing and categorize the packed items into \emph{small, medium,} and \emph{large} items.
	Then, by extending the edges of the large items,
the remaining space is divided into $O(1)$ number of rectilinear regions,
where each region is classified as either \emph{large} or \emph{elongated} or \emph{small}, 
based on the largest item intersecting or completely contained in it.
Then, to make sure that a region has no items partially intersecting it,
they repack all the items as follows.
The {\em large} region is rectangular and its dimensions are much larger
compared to the largest item intersecting it. Here NFDH can repack a high profitable subset of the items.}
An {\em elongated} region is again rectangular and it has one dimension much longer than the other. Here
an algorithm for strip packing  \cite{KenyonR00} is used to repack the items.
Finally, the {\em small} regions are handled by applying the above transformations recursively as they can have complicated shapes (they may not be rectangular).
Recently, Heydrich and Wiese \cite{HeydrichWiese2017} gave an EPTAS, effectively achieving the best-possible approximation.
However, a PTAS for \dsqks{$d$} for $d>2$ remains elusive as it is difficult to find such transformations for higher dimensions.
\ari{Even for 3-D, the  best-known approximation ratio remains 9/8 \cite{harren-journal}. } In 2-D,  many structural theorems \cite{bansal2009structural, GalvezGHI0W17} show that a near-optimal solution exists where all items are packed into $O(1)$-number of rectangular regions where items are either packed as a stack or packed using NFDH. But these results do not generalize to higher dimensions.
E.g., while extending the approach of {\em large} and {\em elongated} blocks of \cite{jansen-journal} to $d>2$, we may not obtain a near-optimal structure packed in $d$-D hypercuboids; rather, we might obtain complicated rectilinear regions. However, prior to our work, there exist no algorithms which considered packing in such rectilinear regions.

\subsection{Our Contributions}
For the \dsqks{$d$} problem, we design a PTAS, thus settling the problem. 
Our main structural result intuitively says that any packing, by incurring a loss of only $\eps$-fraction of the profit, can be transformed into $O(1)$ number of special hypercuboidal regions such that each region either contains a single large item, or contains items very small compared to its dimensions (\Vbox, see \cref{fig:v-box}), or contains items placed along a multidimensional grid (\Nbox, see \cref{fig:n-box}). We then provide an algorithm that guesses the arrangement of these hypercuboids and \mar{uses results of \cite{jansen-journal, assignment} to find} a near-optimal packing of items in these special boxes.

\begin{figure}[h]
\begin{subfigure}{0.4\linewidth}
\centering
\includesvg[width=0.7\linewidth]{img/v-box}
\caption{\kvn{An \Vbox~is just a huge cuboid when compared to its size parameter $\hat s$. Each item has side length at most $\hat s$ and the item set assigned
to it can be packed by NFDH.}}
\label{fig:v-box}
\end{subfigure}
\hspace*{5pt}
\begin{subfigure}{0.5\linewidth}
\centering
\includesvg[width=1.0\linewidth]{img/n-box}
\caption{An \Nbox~is a multidimensional grid. An item set assigned to an \Nbox{} can be packed by placing each item in a cell.}
\label{fig:n-box}
\end{subfigure}
\caption{\Vbox{} and \Nbox{}}
\end{figure}

\arir{Rewrote this subsection completely, please proofread carefully.}

For our structural theorem, we start with an optimal solution. 
First, we divide the items in the optimal packing into large, medium, and small items such that the medium items have a tiny profit and can be discarded.
Then, by extending the facets of the large items, we create a non-uniform grid with $O(1)$
hypercuboidal cells. 
However, some items may intersect the facets of these cells. 
We consider the facets of all the cells and identify (based on the intersecting items) each of the facets as either \emph{good} or
\emph{bad}. 
Intuitively, if we only have good facets then we can repack items so that no items intersect the facets. 
To get rid of the bad facets, we merge some cells so that they form composite cells,
which we call \emph{collections}  (See \cref{3d-blocks-motivation}).
Now we need to repack the items in the {\em collections} which may have complex rectilinear shapes. 
If the collection has $k \in \{0, 1, \dots, d\}$ long dimensions then we call it a $k$-elongated collection. 
Note that we can have $d+1$ types of collections.

Our structural theorem can be viewed as a generalization of
\cite{jansen-journal, harren-journal}; however, due to the complicated shape of collections, we had to overcome several technical obstacles in the process. For simplicity, we explain the intuition for $d=3$, where we obtain a packing of items into four types of rectilinear regions (collections) as in \cref{3d-blocks-motivation}.

\begin{figure}[h]
\centering
\includesvg[width=0.8\linewidth]{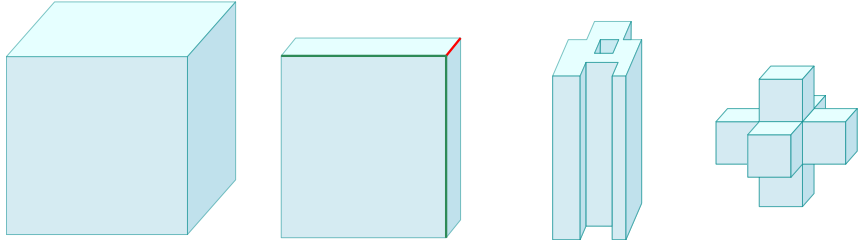}
\caption{Repacking items inside the left-most collection is easy (using NFDH). \kvn{The right-most collections are handled recursively.} However, repacking
	the items in the other two collections is challenging (we needed a novel strip packing algorithm).}
\label{3d-blocks-motivation}
\end{figure}

The first collection (3-elongated) is a \Vbox, i.e., all its three dimensions are long compared to the largest item packed in it. The simple NFDH algorithm suffices to repack the items inside it \mar{while losing only a few items with a small profit}.
Although repacking in large collections can be easily handled by NFDH, there exist no previous algorithms to repack items in the remaining
 types of collections.
 Now for $k\in[d-1]$, 
 collections are long in at least one dimension and thus can be viewed as a type of strip (though the base may not be rectangular).
 Harren \cite{harren-journal} showed how to repack items almost optimally in a $d$-dimensional strip with a rectangular base with
a bounded aspect ratio. However, their approach doesn't trivially extend to the 1- and 2-elongated collections in
\cref{3d-blocks-motivation}, as the base can be non-rectangular or it may not satisfy the property of bounded aspect ratio even though it is rectangular.
In this paper, we circumvent this bottleneck by designing a PTAS (\ari{under resource augmentation}) for $d$-dimensional strip packing to pack items on such complicated bases. Now we explain how we use the strip packing algorithm to pack these collections.


In the second collection (2-elongated), one dimension (marked in red) is reasonably short compared to the largest item \kvn{intersecting} it and the other dimensions (in green) are long.
We first divide the items \ari{in the collection} into large, medium, and small items and ensure that the total volume of medium items is marginal.
We then consider the large items and round their sizes to up to $O(1)$ types using linear grouping \cite{KenyonR00}.  
The number of (rounded) large items that can fit on along the short dimension is $O(1)$.
So, similar to \cite{KenyonR00}, we solve an LP that represents the fractional strip packing with (rounded) large items where the items
are allowed to be sliced along 
\kvnr{along or perpendicular to?} the long dimensions. We then convert this into an integral packing of
 large items using resource augmentation along the long dimensions. Then we pack as many small items as possible in the gaps left. The rest of the small items, together with the medium items, are packed on the top.

The third collection (1-elongated), though complicated, is ensured (by our construction) to have a base with a bounded aspect ratio\footnote{More accurately, the minimal cuboid, which completely contains the base, has bounded-aspect ratio.}.
Moreover, it can be viewed as a union of
$O(1)$ number of smaller disjoint rectangles (which we call \emph{base cells}).
We again classify each item \ari{in this collection} as large, medium, or small. As above, we round the large items to $O(1)$ types.
Note that the number of large items that can fit on the
base of the strip is only a constant (due to the bounded \mar{aspect ratio}\kvnr{same here}). We pack them integrally by first
solving the fractional LP where items are allowed to be sliced along the long dimension.
Our main novelty, in this case, is \kvn{the way in which} we pack the small items on the top of the strip.
Unlike \cite{harren-journal}, we can't use NFDH directly to pack the small items as the base isn't rectangular.
Instead, we first merge a few base cells to form a larger base cell which is
rectangular (we show that this is indeed possible) and \mar{is large enough to accommodate the largest item (and thus any item)}.
Now, having at least one large base cell, we distribute the small \mar{and medium items}
among all the base cells so that when they are packed \mar{using NFDH} in strips on these base cells,
the maximum height is minimized.
Observe that if we pack the items on a single base cell, we may not efficiently use the entire area of the base and this can lead to a very tall strip.
Thus this \kvn{process of distributing the items among all the base cells} is important.
\kvnr{For the second and third collections, we haven't quite explained how the $N$-boxes and $V$-boxes are formed in the process.}

We use our near-optimal strip packing algorithm on these \mar{1- or 2-elongated} collections to repack the items while losing only a small profit.
\marr{Moved the first sentence of this paragraph down from the paragraph above.}
Interestingly, the strip packing algorithm for both kinds of bases (\ari{in general, for any $k$-elongated collection where $k \in [d-1]$}) is the same; the only change is the number of
short dimensions of the base that is given as a parameter to the algorithm. We provide the
algorithm for $d$-dimensional strip packing in \cref{sec:strip packing}.

Finally, for 0-elongated collections, we apply the above process recursively.
Using a shifting argumentation, we show that $O_\eps(1)$ steps of recursion is sufficient.
We also show that the resulting packing is a packing into $O(1)$ number of \Vboxes{} and \Nboxes{}.
With more
technical ingenuity, these ideas can be extended to the case of $d>3$.

\arir{removed unproven claims.}
\ari{We believe that our techniques will be helpful for other multidimensional geometric packing problems involving rectilinear regions that are not hypercuboids.}


\arir{used \Nbox instead of $N$-Box as we use $N$ for number of cells.}
\subsection{Related Work}
For geometric bin packing, Caprara gave a  $1.691^{d-1}$ asymptotic approximation 
\cite{caprara2008packing}.
For $d=2$, it admits no APTAS  \cite{bansal2006bin} and there is a 1.406-approximation algorithm \cite{bansal2014binpacking}.
There are several other variants of bin packing such as vector packing \cite{bansal2016improved, sandeep2021almost}, strip packing \cite{harren20145, Galvez0AJ0R20, JansenR19}, sliced packing \cite{DeppertJ0RT21, Galvez0AK21}, mixed packing \cite{KSS21bo}, etc.
The knapsack problem is also studied under several generalizations such as vector knapsack \cite{frieze1984approximation}, fair knapsack \cite{PKL21}, and mixed knapsack \cite{KSS21}.
Packing problems are also well-studied under guillotine cut constraints \cite{BansalLS05, KMSW21, KMR20, AbedCCKPSW15, khan2021tight}. Guillotine cut also has interesting connections with the maximum independent set of rectangles problem \cite{AHW19, Galvez22}.
We refer the readers to  \cite{CKPT17, khan2015approximation} for a survey on multidimensional packing. 
\subsection{Overview of the Paper}
\ari{In \cref{sec:preliminaries}, we present some preliminaries. \cref{sec:knapsack} and \cref{sec:strip packing} contain the results for the \dsqks{$d$} and \dsqsp{$d$}, respectively. 
Due to space limitations, many proofs had to be moved to the Appendix.}

\section{Notations and Preliminaries}\label{sec:preliminaries}
Let $[n] := \Set{1,2,\dots,n}$.
For any set of items $\mathcal{I}$,
we define $\hat s(\mathcal{I}) := \max_{i\in\mathcal{I}} s(i)$.
For any $k$-dimensional region $R$, we denote its volume by $\vol[k]{R}$. 
We extend this notation to an item $i$ or a set of items $\mathcal{I}$, i.e.,
$\vol[d]{i} = (s(i))^d$, 
and $\vol[d]{\mathcal{I}} = \sum_{i \in \mathcal{I}}\vol[d]{i}$.
Also, the profit of a packing is the sum of the profits of the packed items $\mathcal{Q}$: $p\paren{\mathcal{Q}} := \sum_{i\in \mathcal{Q}}p\paren{i}$.
  
Consider a set of points $X$ in $[0,1]^d$ where each point
$x\in X$ can be represented as $(x_1,x_2,\dots,x_d)$. For any set of dimensions
$\mathcal D \subseteq [d]$, we define the projection of $X$ onto $\mathcal D$
as the set $\Set{\paren{x_{d_1}, \dots, x_{d_{\abs{\mathcal D}}}} | \paren{x_1, \dots, x_k} \in X}$,
where $\Set{d_1, \dots, d_{\abs{\mathcal D}}} = \mathcal D$ and $d_i < d_{i+1}$ for all $i \in [\abs{\mathcal D}-1]$.

We only consider axis-aligned packing of hypercuboids and hypercubes.
Thus, a $d$-D hypercuboid $C$ is given by the position of its lower corner $\paren{x_1, \dots, x_d}\in\real^d$
and side lengths $\ell_1, \dots, \ell_d \in\real_+$.
Then, $C$ is the cartesian product of the intervals $\prod_{i=1}^d \intv{x_i, x_i + \ell_i}$, and $\vol{C} := \prod_{i=1}^d \ell_i$
and its surface area is defined as the sum of the volumes of its $\paren{d-1}$-dimensional facets, i.e.,
$\surf{C} := 2 \sum_{i=1}^d \prod_{j=1, j \neq i}^d \ell_j = 2 \sum_{i=1}^d \vol{C}/\ell_i$.
A packing is a subset $\mathcal{Q} \subseteq \mathcal{I}$ of items
with positions $\operatorname{pos} : \mathcal{Q} \longrightarrow \left[0,1\right]^d$.
It is valid, if each item $i \in \mathcal{Q}$ with the lower corner positioned at $\operatorname{pos}\paren{i}$ is completely included in the unit hypercube
$[0,1]^d$ and each pair of items $i,j \in \mathcal{Q}$ is positioned non-overlapping.

\kvn{
  Consider a set of items lying in $d$ dimensional space. For any region in the $d$ dimensional space,
  we define the profit of that region as the profit of the set of items \emph{completely} contained
  in that region.
}

\kvn{
  We will use the higher dimensional variant of the well-known Next Fit Decreasing Height (NFDH)
  algorithm extensively. For the sake of completeness, we provide the details of the algorithm
  in \cref{app:nfdh}.
}
Harren \cite{harren-journal} gave a surface area based efficiency guarantee of NFDH
algorithm.
\begin{lemma}[\cite{harren-journal}]\label{nfdh}
Consider a set $S$ of hypercubes with side lengths at most $\delta$ and a hypercuboid $C$.
The NFDH algorithm either packs all items from $S$ into $C$
or the total volume left free inside of $C$ is at most $\delta \surf{C}/2$.
\end{lemma}

Now we define \Vbox{} and \Nbox{}.
\arir{better to use $\mathcal{I}$ only for input items. For others, we can use $\mathcal{J}$.
Changed.}
\begin{definition}[\Vbox~and \Nbox]
Let $B$ be a $d$-dimensional hypercuboid, $\mathcal{Q}$ be a set of
items packed in it, and let $\hat s$ be an upper bound on the
side lengths of the items in $\mathcal{Q}$.
We say that $B$ is a \Vbox{} if its side lengths can be written
as $n_1\hat s, n_2\hat s,\dots, n_d\hat s$ where
$n_1,n_2,\dots,n_d\in \mathbb{N}_+$ and the volume of the packed items is at most 
$\vol{B}-\hat s\frac{\surf{B}}{2}$.
We say that $B$ is an \Nbox{} if its side lengths can be written
as $n_1\hat s, n_2\hat s,\dots, n_d\hat s$ where
$n_1,n_2,\dots,n_d\in \mathbb{N}_+$ and the number of
packed items is at most $\prod_{i=1}^d n_i$.
\end{definition}
\kvnr{Certainly needs some changes. What is $\hat s$ in the definition. For both n-box and v-box we should say that each item has size at most $\hat s$}
As the side lengths are integral multiples of $\hat s$, the number of distinct
\Vboxes{} and \Nboxes{} in the near-optimal structure will be polynomially bounded.
\kvn{The following corollary follows from \cref{nfdh}}.
\begin{corollary}\label{v-box}
Let $B$ be a \Vbox{} with the item set $\mathcal{Q}$ packed in it. Then these items can be repacked into $B$ using NFDH.
\end{corollary}
\begin{proof}
\kvn{We can either pack all the items in $\mathcal I$ into $B$ using NFDH, or, by \cref{nfdh},}
the free volume inside of $B$ is at most $\hat{s} \surf{B}/2$.
The latter possibility however implies that we packed items with a volume of at least $\vol{B}-\hat{s}\surf{B}/2$,
which is an upper bound on
\kvn{$\vol{\mathcal I}$}
by the definition of a \Vbox{}.
\end{proof}
\begin{observation}\label{n-box}
Let $B$ be an \Nbox{} with the item set $\mathcal{Q}$ packed in it.
As $\abs{\mathcal{Q}}\le \prod_{i=1}^d n_i$, we can divide $B$ into
$\prod_{i=1}^d n_i$ cells such that each cell has length $\hat s$ in each of the $d$
dimensions and we can place each item in $\mathcal{Q}$ in one cell.
\end{observation}
In this paper, we will often require that a set of hypercuboidal regions forms a grid.
\begin{definition}\label{grid:def}
A \emph{$d$-dimensional grid} (see  \cref{fig:preliminaries:grid:splitting}(a)) $C$ is a subset of the set of hypercuboids
$\Set{\prod_{i=1}^d\left[g_{i,j_i-1}\;,\; g_{i,j_i}\right] | 1\leq j_i \leq n_i\text{ for all }i \in [d]}$, 
where each $n_i \in \nat_+$ denotes the number of grid layers in the $i\Th$ dimension, and
$g_{i, 0} < \dots < g_{i, n_i}$ for all $i\in[d]$ are the grid boundaries.
Each hypercuboid in $C$ is called a \emph{cell}.
\end{definition}
Each grid $C$ has $\abs{C} \leq \prod_{i=1}^d n_i$ cells.
A grid $C$ can also be refined by splitting some cells.

\begin{lemma}\label{grid:splitting}
Let $C$ be a $d$-D grid with $n_1, \dots, n_d \in \mathbb{N}_+$ layers and boundaries $g_{i,j} \in \mathbb{R}$ for $i \in [d]$ and $j \in [n_i]$.
Let $R$ be a set of $d$-D hypercuboids.
Then the region of the grid which is not intersected by $R$ is a new grid with
$n'_i \leq n_i + 2\left|R\right|$ layers in each dimension $i\in [d]$.
\end{lemma}
\begin{proof}
Each hypercuboid in $R$ is the cartesian product of $d$ intervals $\prod_{i=1}^d \left[a_i, b_i\right]$.
The grid $C$ can now be refined by adding the boundaries $a_i$ and $b_i$ in each dimension $i$.
This increases the number of layers in each dimension by at most $2\left|R\right|$.
After this refinement,
each cell is either completely contained in a hypercuboid in $R$,
and hence discarded from $C$, or does not intersect $R$.
This can be seen in \cref{fig:preliminaries:grid:splitting}(b).
\end{proof}

\begin{figure}[h]
  \centering
  \begin{subfigure}{.4\textwidth}
    \centering
    \includesvg[width=\linewidth]{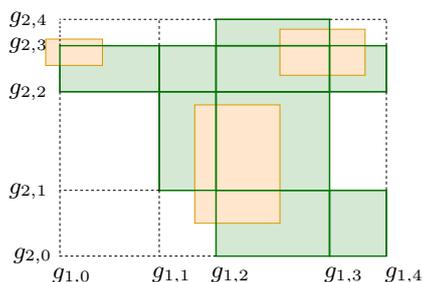}
     \caption{3 orange rectangles partially covering a grid with 4 layers in each dimension.\kvn{The dotted lines indicate the grid boundaries and the
            grid cells are indicated by the green rectangles.}}
  \end{subfigure}
  \hspace{1cm}
  \begin{subfigure}{.4\textwidth}
    \centering
    \includesvg[width=0.8\linewidth]{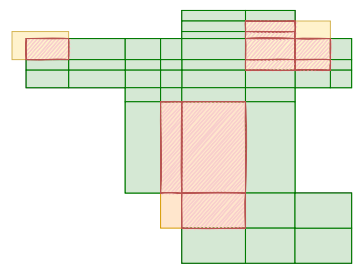}
    \caption{Refined grid with $\le 10$ layers in each dimension.
    The new grid cells are indicated in green. The shaded cells are those discarded from the refined grid as they overlap with the items.}
  \end{subfigure}
  \caption{Splitting the empty space inside of a grid}
  \label{fig:preliminaries:grid:splitting}
\end{figure}

\section{Knapsack}\label{sec:knapsack}
\ari{
  In this section, we will devise a PTAS for the \dsqks{$d$} problem.} 

\subsection{Structure of a Nearly Optimal Solution}\label{sec:knapsack:structure}

First, we prove that
given an optimal packing $\optk{\mathcal{I}}$ of the input set $\mathcal{I}$, there exists another packing
containing a subset of the items packed in a simple structure
which can be searched for, in polynomial time. \kvn{For this,
we modify the optimal packing to obtain a near-optimal packing in which all the
items are packed into either \Vboxes{} and \Nboxes{} except for a constant number of large items. The total number of
these boxes is $\mathcal{O}(1)$ and their sizes come from a set whose cardinality is
polynomial in $\abs{\mathcal I}$. Hence, this near-optimal packing can be
searched for, in polynomial time.}

\begin{theorem}\label{knapsack:structure}
For each \ari{$0< \epsilon <1/2^{d+2}$},  there is a packing with the following properties:
\begin{enumerate}[(i)]
\item It consists of \Nboxes{} and \Vboxes{} whose total number is bounded by
a constant $\Cname{boxes}\paren{d, \epsilon}$, which depends only on $\epsilon$ and $d$.
\item The number of items in the packing that are not packed in these boxes is bounded
by a constant $\Cname{large}\paren{d, \epsilon}$, which depends only on $\epsilon$ and $d$.
\item The profit of the packing is at least $\paren{1-2^{d+2}\epsilon}\optk{\mathcal{I}}$. \arir{add some upper bound on $\eps$?}
\end{enumerate}
\end{theorem}
To prove the theorem, first, in
\cref{sec:knapsack:structure:classification}, we consider a packing in an arbitrary grid
(note that the knapsack can be viewed as a grid
containing a single cell) and merge some
of its cells into collections based on the packing of the items into the grid.
\ari{
\kvnr{The reason for choosing an arbitrary grid is that, as we will see, ---- Can we retain this? retained}
The reason for choosing an arbitrary grid is that, as we will see,
we recursively divide the knapsack
into grids (which can have arbitrary shapes) and consider each of them separately.}
Then in
\cref{sec:knapsack:structure:repacking}, we show how to repack the items in a collection
using NFDH or by a strip packing algorithm \mar{described in \cref{sec:knapsack:structure:strip packing}.}

\subsubsection{Partitioning a Grid into Collections}\label{sec:knapsack:structure:classification}
Consider a grid and a set of items $\mathcal J$ packed into it. Assume that
the grid has at most $\Nname{layer}$ layers in each dimension.
Thus, it consists of at most $N := \Nname{layer}^d$ cells.
Our goal is to repack \ari{(a subset of ) $\mathcal J$ to obtain a simpler structure, by losing only a small profit}. The overall grid can have a complex structure, so we consider each cell (a hypercuboid) and repack the items in the cell. However, it might be problematic to repack the items that intersect a cell only partially.
\kvn{
  Consider one such item and one of the cells it intersects. There must be a facet of this cell that cuts
  into this item.
}
If this facet is orthogonal to a long dimension \kvn{of the cell}, then we can remove a strip of small profit
and pack these intersecting items in that strip. However, an issue arises when this facet is
orthogonal to a short dimension. Therefore,
we merge the two cells that share this facet. Doing this iteratively,
we merge some cells into a \emph{collection of cells} or simply \emph{collection}. See \kvn{left of} \cref{fig:knapsack:structure:merging}.

For a cell $a$, we will denote the side length of the largest item of $\mathcal{J}$ partially \kvn{or fully} packed in it by $\hat{s}\paren{a}$
and we denote the side lengths of the cell by $a_1, \ldots, a_d$.
Furthermore, we sort the dimensions using a stable sorting algorithm \footnote{\mar{A stable sorting makes the order unambiguous. This simplifies some proofs that compare the orders for different collections.}},
\arir{why do we need stable sorting here?}
\marr{A stable sorting makes the orders for different collections easier to compare when we have multiple dimensions with the same size.
E.g. the equalities for the order of dimensions in the proof of \cref{knapsack:structure:merging} break when there are multiple dimensions with the same length that are ordered differently for the different collections.
Using a non-stable sorting would require us to do some workaround there.
\ari{Added this as footnote. Please check and edit.}}
such that the side lengths of $a$ in those dimensions are in non-increasing order.
We will refer to that order as $\sigma_a$ and thus, $a_{\sig{a}{i}}$ is the $i\Th$ largest side length of $a$.
\kvnr{Should we shift this paragraph to the notations section?}
We start by defining the values which we will use to distinguish between long and short sides of a cell as
\begin{align}
\alpha_d := \frac{2d}{\epsilon},\text{ and }
\alpha_{k} := \frac{2}{\epsilon}\paren{\Cname{configs}\paren{d, k, \alpha_{k+1}, N, \epsilon}+3}\text{ for all } k\in[d-1]\label{alpha-defn}
\end{align}
The parameter $\Cname{configs}\paren{d, k, \alpha_{k+1}, N, \epsilon}$ is a constant depending on
$d,k,\alpha_{k+1},N,\eps$. Its exact value is derived in \cref{sec:strip packing}, but for our purposes, it suffices to
note that for all $k\in[d-1]$, $\Cname{configs}\paren{d, k, \alpha_{k+1}, N, \epsilon}\ge 1$ and its value gets larger when $k$ gets smaller.
Hence, $1<2d/\eps = \alpha_d < \ldots < \alpha_1$.

Now we can define a cell $a$ to be $k$-elongated for  $0 \leq k \leq d$
iff $k$ of its side lengths are long compared to some size parameter $s$.
\begin{definition}\label{knapsack:structure:k-elongated:cell}
A cell $a$
is $k$-elongated for $k \in [0, d]$ and size parameter $s > 0$ iff
\begin{enumerate}[(i)]
\item $a_{\sig{a}{i}} > \alpha_k s$ for all $1 \leq i \leq k$ and
\label{knapsack:structure:k-elongated:cell:a}
\item $a_{\sig{a}{i}} \leq \alpha_i s$ for all $k < i \leq d$.
\label{knapsack:structure:k-elongated:cell:b}
\end{enumerate}
\end{definition}
\kvn{If $a$ is a $k$-elongated cell,
then for any $i\in[k]$ we call $\sig{a}{i}$ to be a \emph{long} dimension of $a$
and the other dimensions to be \emph{short}.}
Note that for each size parameter, each cell is $k$-elongated for exactly one value of $k$.
A facet of a $k$-elongated cell is called \emph{\narrow{}} if it is orthogonal to a long dimension and \emph{bad} otherwise.
\kvn{We will see that the items that intersect only the good facets of a cell don't cause much of an issue;
the complicated machinery that we devise is to deal with the items intersecting the bad facets.
}

\kvn{For a group of cells to be merged into a collection, they should satisfy some additional properties.}
We define a \emph{$k$-elongated collection} $C$ as a set of cells that are $k$-elongated using the
size of the largest item \kvn{completely or partially} packed in that collection as a common size parameter
and are aligned in their long dimensions. \kvn{More formally,}
\begin{definition}\label{knapsack:structure:k-elongated:collection}
A set of cells $C$ is a $k$-elongated collection if
\begin{enumerate}[(i)]
\item each cell $a\in C$ is $k$-elongated for the size parameter $\hat{s}\paren{C} := \max_{c\in C}\paren{\hat{s}\paren{c}}$,
\label{knapsack:structure:k-elongated:collection:a}
\item $\sig{a}{i} = \sig{b}{i}$ for all $1 \leq i \leq k$ and all $a,b \in C$ and
\label{knapsack:structure:k-elongated:collection:b}
\item the projection of $a$ and $b$ onto the dimension $\sig{a}{i}$ is the same for all $1 \leq i \leq k$ and all $a,b \in C$.
\label{knapsack:structure:k-elongated:collection:c}
  \item The set of cells form a \kvn{path-}connected region.
\end{enumerate}
\end{definition}
\begin{figure}
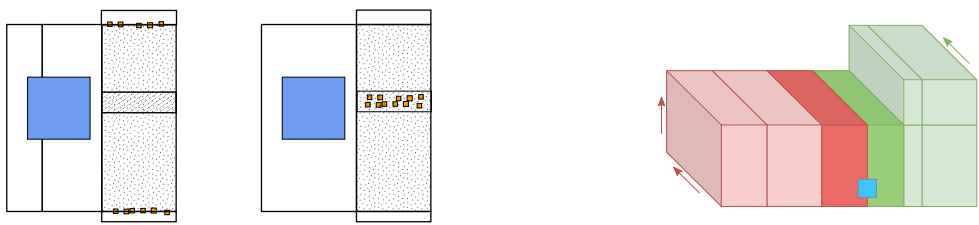

  \begin{subfigure}{0.5\textwidth}
      \hspace*{0.5cm}
      \includesvg[width=0.8\linewidth]{img/merge-cells}
  \end{subfigure}
  \begin{subfigure}{0.5\textwidth}
      \hspace*{2cm}
      \includesvg[width=0.6\linewidth]{img/merging}
  \end{subfigure}
  \caption{
  \kvn{\emph{(i)} The left figure} shows \kvn{five} cells of a grid \ari{in the  two-dimensional case}. The orange items are  small compared to the height of \kvn{one of the cells (shaded with dots)} they intersect. So, we transfer them to the interior of
      that cell by removing a least profitable strip \kvn{(shaded in stripes)}. However, the blue item cannot fit in any of the cells
      it intersects, so we merge both these cells.
      \kvn{\emph{(ii)}~The right figure demonstrates merging.}
      \mar{An item (blue) intersects a \wide{} facet between a $2$-elongated collection (red)
      and a $1$-elongated collection (green).
      \kvn{
      After merging, we obtain a single $1$-elongated collection.}
      }
  }
  \label{fig:knapsack:structure:merging}
\end{figure}
\begin{figure}[h]
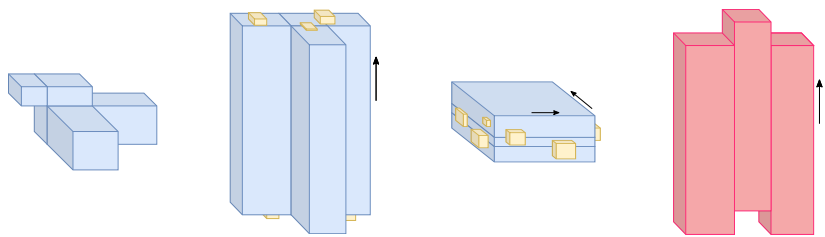

  \centering
  \begin{subfigure}{0.20\textwidth}
    \centering
    \includesvg[width=0.7\linewidth]{img/knapsack_0_elongated}
  \end{subfigure}
  \begin{subfigure}{0.20\textwidth}
    \centering
    \includesvg[width=0.7\linewidth]{img/knapsack_1_elongated_items}
  \end{subfigure}
  \begin{subfigure}{0.20\textwidth}
    \centering
    \includesvg[width=0.7\linewidth]{img/knapsack_2_elongated_items}
  \end{subfigure}
  \begin{subfigure}{0.20\textwidth}
    \centering
    \includesvg[width=0.7\linewidth]{img/non-collection}
  \end{subfigure}
  \caption{Different collections in 3-D with some items intersecting good facets. The arrow marks indicate
  the long dimensions. From left to right: $0$-, $1$-, $2$-elongated collections, and a non-example of a collection (violates property (\ref{knapsack:structure:k-elongated:collection:c}).)}
  \label{fig:knapsack:structure:collections}
\end{figure}
By properties (\ref{knapsack:structure:k-elongated:collection:a}) and (\ref{knapsack:structure:k-elongated:collection:b}) of a $k$-elongated collection, we can define the $k$ dimensions that are long for any cell of the collection to be the long dimensions of the collection itself.
Property (\ref{knapsack:structure:k-elongated:collection:c}) strictly limits the arrangement of the cells in long dimensions.
Thus, a larger $k$ makes the shape of the whole collection less complex (See \cref{fig:knapsack:structure:collections}).
A common facet between two cells $a\in C$ and $b\in D$ of two different collections $C$ and $D$ is called an outer facet for each of those collections.
\mar{For each dimension, each cell has two facets that are orthogonal to this dimension.
The facet with the lower coordinate in this dimension is called the bottom facet and the other one is called the top facet.
For each long dimension of a collection the {\em bottom facet of the collection} is formed by the union of the bottom facets of its cells; the {\em top facet of the  collection} is formed by the top facets of its cells.}

Initially, we consider each cell of our grid as a collection containing only itself.
Whenever there is an item intersecting an outer facet between two collections $C$ and $D$,
we will merge those collections if that facet is \wide{} (i.e. orthogonal to a short dimension) for both of them.
\mar{Such a case is visualized \kvn{on the right of} \cref{fig:knapsack:structure:merging}.}
The next lemma proves that merging collections in this case will create a new collection.
\begin{lemma}\label{knapsack:structure:merging}
Let $C$ be a $k_C$-elongated collection and let $D$ be a $k_D$-elongated collection.
Let $a\in C$ and $b\in D$ be two cells that have a common facet which is \wide{} for both of them.
Then $C\cup D$ is a $\min\paren{k_C,k_D}$-elongated collection.
\end{lemma}
\begin{proof}
Note that both $k_C,k_D$ are strictly less than $d$ since $a,b$ share a common facet $f$ which is \wide{} for both.
W.l.o.g., we will assume that $\hat{s}\paren{C} \geq \hat{s}\paren{D}$.
Let $f$ has side lengths $f_1, \ldots, f_d$ \kvn{where $f_i$ denotes
the length in the $i\Th$ dimension} (for simplicity, we assume $f$ is a
$d$-dimensional hypercuboid with a side of zero length).
Note that the facet $f$ has the same side lengths as $a$ and $b$ in any dimension except the dimension in which \kvn{it has zero length}.
This implies that $a$ and $b$ are of almost the same shape.
In the dimension orthogonal to $f$, both $a$ and $b$ have a rather short length,
as $f$ is a \wide{} facet for both collections.
Using these arguments, we can derive the following properties:
\begin{enumerate}[(i)]
\item $a_{\sig{f}{i}} = f_{\sig{f}{i}} = b_{\sig{f}{i}}$ for all $1\leq i \le d-1$\hfill{{(Since $f$ is common to both $a,b$)}}
\label{knapsack:structure:merging:1}
\item $\sig{f}{i} = \sig{a}{i}$ for all $1\leq i \leq k_C$\hfill{{(Since $f$ is \wide{} for $a$ which is a $k_C$-elongated cell)}}
\label{knapsack:structure:merging:2}
\item $\sig{f}{i} \in \Set{\sig{a}{i}, \sig{a}{i+1}}$ for all $k_C < i < d$\hfill{{(explained below)}}
\label{knapsack:structure:merging:3}
\item $\sig{f}{i} = \sig{b}{i}$ for all $1\leq i \leq k_D$
\label{knapsack:structure:merging:4}
\item $\sig{f}{i} \in \Set{\sig{b}{i}, \sig{b}{i+1}}$ for all $k_D < i < d$
\label{knapsack:structure:merging:5}
\end{enumerate}
  The property (iii) above is due to the fact that $a$ and $f$ have the same lengths in all dimensions except one, say $d_\perp$. So, the
  order of dimensions $\sig{a}{1},\sig{a}{2},\dots,\sig{a}{d}$ can be obtained by inserting $d_\perp$ in the list
  $\sig{f}{1},\sig{f}{2},\dots,\sig{f}{d-1}$ at the appropriate index but we know that this index lies after $k_C$ since $d_\perp$
  is short for $a$.

If we assume $k_C > k_D$, we obtain the inequality
$f_{\sig{f}{k_C}} = a_{\sig{a}{k_C}} > \alpha_{k_C}\hat{s}\paren{C}$ by (i), (ii) and \cref{knapsack:structure:k-elongated:cell} of the $k_C$-elongated cell $a$.
We also obtain the contradictory inequality $f_{\sig{f}{k_C}}=b_{\sig{f}{k_C}} \le b_{\sig{b}{k_C}} \leq \alpha_{k_C}\hat{s}\paren{D} \leq \alpha_{k_C} \hat{s}\paren{C}$ by (i), (iv) {and (v)} and \cref{knapsack:structure:k-elongated:cell} of the $k_D$-elongated cell $b$.
Thus, our choice $\hat{s}\paren{C} \geq \hat{s}\paren{D}$ determines that $k_C \leq k_D$.

As $\hat{s}\paren{C\cup D} = \hat{s}\paren{C}$ and $C$ is a $k_C$-elongated collection,
we know that $C\cup D$ fulfills property (\ref{knapsack:structure:k-elongated:collection:a}) of \cref{knapsack:structure:k-elongated:collection} for a $k_C$-elongated collection for every $a' \in C$.
Let $b' \in D$. We have to show that $b'$ is a $k_C$-elongated cell using the size parameter $\hat{s}\paren{C}$ as well.
For every $1\leq i \leq k_C$ we can use \cref{knapsack:structure:k-elongated:collection} of the $k_D$-elongated collection $D$ and (\ref{knapsack:structure:merging:1}), (\ref{knapsack:structure:merging:4}) and (\ref{knapsack:structure:merging:2}) to prove
$b'_{\sig{b'}{i}} = b_{\sig{b}{i}} = f_{\sig{f}{i}} = a_{\sig{a}{i}} > \alpha_{k_C}\hat{s}\paren{C}$.
For every $k_C< i \leq k_D$ we can again use \cref{knapsack:structure:k-elongated:collection} of the collection $D$ and properties (\ref{knapsack:structure:merging:1}), (\ref{knapsack:structure:merging:4}) and (\ref{knapsack:structure:merging:3}) to prove
$b'_{\sig{b'}{i}} = b_{\sig{b}{i}} = f_{\sig{f}{i}} \leq a_{\sig{a}{i}} \leq \alpha_{i}\hat{s}\paren{C}$.
For every $k_D< i \leq d$ we already know that $b'_{\sig{b'}{i}} \leq \alpha_{i}\hat{s}\paren{D} \leq \alpha_{i}\hat{s}\paren{C}$ because of \cref{knapsack:structure:k-elongated:cell} for the $k_D$-elongated cell $b'$.
Thus, $b'$ is a $k_C$-elongated cell using the size parameter $\hat{s}\paren{C}$.

Now we only need to prove properties (\ref{knapsack:structure:k-elongated:collection:b}) and (\ref{knapsack:structure:k-elongated:collection:c}) of \cref{knapsack:structure:k-elongated:collection} for {every pair of cells in} $C\cup D$ to show that it is a collection.
For a pair $a', a'' \in C$ or $b', b'' \in D$ this is trivial, because $C$ and $D$ are $k_C$- or $k_D$-elongated collections and $k_C \leq k_D$.
For the pair $a$ and $b$, those properties are again trivial, because both have all their long sides on the common facet $f$.
For any $a' \in C$ and $b' \in D$ the properties follow as the pairs $a'$ and $a$, $a$ and $b$, $b$ and $b'$ satisfy these properties.
\end{proof}

After this merging procedure, we will assign each item $x\in\mathcal{J}$ to a collection $C$ as follows:
If $x$ is completely packed in a collection, then we assign it to that collection itself.
If it is partially contained, then we assign it to one of the collections
(breaking ties arbitrarily) with which
it intersects only via the \narrow{} outer facets.
The next lemma shows that this assignment is indeed possible.

\begin{lemma}\label{knapsack:structure:assignment}
Let $G$ be the set of all collections that were derived after the merging procedure.
Let $x\in\mathcal{J}$ be an item that is packed in the grid \kvn{but isn't completely
packed in any collection in $G$.}
Then there exists a collection $C\in G$ in which $x$ is partially packed
but $x$ does not intersect any of its \wide{} outer facets.
\end{lemma}
\begin{proof}
First, we define an order on the set of the collections $G$.
We associate each $k_C$-elongated collection $C\in G$ with the tuple of the long sides $\paren{a_{\sig{a}{1}}, \ldots, a_{\sig{a}{k_C}}}$ for some cell $a\in C$ in a descending order.
By property (\ref{knapsack:structure:k-elongated:collection:c}) in the definition of a $k_C$-elongated collection, the tuple is independent of the choice of the cell.
We now define the strict lexicographic order $\prec$ on these tuples.
In other words, for a $k_C$-elongated collection $C$, a $k_D$-elongated collection $D$ and cells $a\in C$ and $b \in D$, we have $C \prec D$ iff
\begin{enumerate}[(i)]
\item $k_C < k_D$ and $a_{\sig{a}{i}} = b_{\sig{b}{i}}$ for all $1 \leq i \leq k_C$, or
\label{knapsack:structure:assignment:1}
\item there is some $1 \leq k \leq \min\paren{k_C, k_D}$ such that
$a_{\sig{a}{i}} = b_{\sig{b}{i}}$ for all $i\in[k-1]$ and
$a_{\sig{a}{k}} < b_{\sig{b}{k}}$.
\label{knapsack:structure:assignment:2}
\end{enumerate}
Let $H \subseteq G$ be the set of collections where $x$ is partially contained in.
Let $C \in H$ be a maximal collection in $H$ according to $\prec$ and let it be $k_C$-elongated.
Now we want to prove that $x$ does not intersect any \wide{} outer facet of $C$.
Assume there is a \wide{} outer facet $f$ of $C$ which is intersected by $x$.
This facet $f$ belongs to some cells $a\in C$ and $b\in D$ where $D$ is a different collection than $C$.
The facet $f$ has to be a \narrow{} facet of $D$, because it is \wide{} for $C$ and the collections $C$ and $D$ were not merged.
Let $d_\bot$ be the dimension which is orthogonal to $f$.
Like in the proof of \cref{knapsack:structure:merging}, we know that $a_i = b_i$ for all $i \in \Set{1,\ldots,d}\backslash \Set{d_\bot}$ because of their common facet $f$.
As $f$ is \wide{} for $C$ this contains all the sides of $a$ that are long using the size parameter $\hat{s}\paren{C}$.
If a side of $a$ is long using the size parameter $\hat{s}\paren{C}$
and larger than $b_{d_\bot}$, then this side is also long using the size parameter $\hat{s}\paren{D}$,
because $b_{d_\bot}$ is already considered as long using this size parameter.
Thus, we have two cases:
In the first case, each side of $a$ which is long using the size parameter $\hat{s}\paren{C}$ is larger than $b_{d_\bot}$.
Then each long side of $a$ is also a long side of $b$ while $b$ has at least one additional long side $b_{d_\bot}$.
Thus, we have $C \prec D$ by (\ref{knapsack:structure:assignment:1}).
In the second case there are only $k < k_C$ sides of $a$ long using the size parameter $\hat{s}\paren{C}$ and larger than $b_{d_\bot}$.
Then the $k$ longest sides of $a$ and $b$ have the same size and for the $\paren{k+1}\Th$ longest sides we have 
$a_{\sig{a}{k+1}} < b_{d_\bot} = b_{\sig{b}{k+1}}$.
Thus, we have $C \prec D$ by (\ref{knapsack:structure:assignment:2}).
In both cases, the maximality of $C$ is violated, and therefore no \wide{} outer facet of $C$ can be intersected by $x$.
\end{proof}

\subsubsection{Strip Packing with Resource Augmentation}\label{sec:knapsack:structure:strip packing}
Before we proceed to repack the items inside of those collections,
we will have a closer look at an arbitrary $k$-elongated collection for some $k \in [d-1]$.
This constraint on $k$ ensures that the collection is large in at least one dimension and short in at least one dimension.
Thus our collection is some kind of a  strip (may not be hypercuboidal).
Previously, we reordered the dimensions for a collection to have the $k$ long dimensions first;
however, for the strip packing problem we assume the opposite, i.e., we suppose that the collection
is short in the first $(d-k)$ dimensions and long in the last $k$ dimensions.
Each cell of this collection is the cartesian product of some intervals and property (\ref{knapsack:structure:k-elongated:collection:c}) of \cref{knapsack:structure:k-elongated:collection} ensures that the last $k$ intervals are the same for any cell in this collection.
By splitting the cartesian product for a cell after the first $(d-k)$ (short) intervals and before the last interval, we can represent the cell as $a \times B \times [h_1,h_2]$ where $a$ is a $(d-k)$-dimensional hypercuboid that depends on the selected cell and both the $(k-1)$-dimensional hypercuboid $B$ and the values $h_1$ and $h_2$ are common to all cells of the collection.
\kvn{This set of} hypercuboids $\{a\}$ forms a $(d-k)$-dimensional grid $A$ that is the projection of our collection onto the first $(d-k)$ dimensions.
With these definitions, our collection can be viewed as a $d$-dimensional strip with base $\paren{\cup_{a\in A}a}\times B$ and height $(h_2 - h_1)$.
\kvn{
  We would like to repack this collection using a strip packing algorithm. For this purpose,
  we devise an algorithm to pack hypercubes on a base (which can be extended by a
  slight amount in the long dimensions) with the properties described above. With this motivation, we formally define
  the strip packing problem and state the main result in the next few paragraphs.
}




\kvn{
  Let $k\in[d-1]$.
  Let $\eps>0$ be some accuracy parameter and $N\in\mathbb N_+,\alpha\ge 1$ be some constants depending on $d,\eps$.
}
The input consists of a set of items $\mathcal I$.
In the entire subsection, we abbreviate $\hat s(\mathcal{I})$ by $\hat s$.

The $(d-1)$-dimensional base on which we need to pack the input set $\mathcal{I}$ is the cartesian product $A \times B$ where
$A$ is some $(d-k)$ dimensional grid of at most $N$ cells and $B$ is some $\paren{k-1}$ dimensional hypercuboid.
The side length of $B$ in the $i\Th$ dimension ($i\in[k-1]$) is denoted by $b_i$.
Each side length of each cell in $A$ has to be at most $\alpha \hat{s}$
and each $b_i$ has to be at least $N\alpha^{d-k} \hat{s}$.
As it can be seen, compared
to $\hat s$, all the edge lengths in $A$ are short and all the edge lengths in $B$ are long. Hence, we refer to the dimensions of the coordinate system parallel to
the edges of $A$ as \emph{short dimensions} and to the dimensions parallel to
the edges of $B$ as \emph{long dimensions}.
In the $3$-dimensional setting, we have two possible values for $k$.
If $k=1$, then $B$ vanishes and our base is just a grid of rectangles $A$.
Such a base can be seen in \cref{fig:strip packing:2d base:1}.
If $k=2$, then $B$ is an interval and $A$ is just a set of non-overlapping intervals because it is $1$-dimensional.
In this case, the base can be visualized like in \cref{fig:strip packing:2d base:2} as a set of flat but wide rectangles, that are positioned on top of each other.
So for the $3$-dimensional case, we either have a complex structure through $A$ or long sides through $B$ but never both.
In higher dimensions however, we can have both.
An example for this is the base of a $4$-dimensional strip in \cref{fig:strip packing:base}.
By scaling
\kvn{every dimension equally,}
\kvnr{Commented out $1/\vol[d-k]{A}^{1/\paren{d-k}}$ scaling factor. I think this is incorrect. Once check.}
it is assumed that the volume of the grid $A$ is normalized: $\vol[d-k]{A} := \sum_{a \in A}\vol[d-k]{a} = 1$.
\kvn{We also assume that any item in $\mathcal I$ can be packed on the base $A\times B$.}
The goal is to find a non-overlapping packing of the items in $\mathcal{I}$ into the region given by the set
product of $A$, $B$ and $h$, that is $\cont{A, B, h} := (\bigcup_{a \in A} a) \times B \times \intv{0, h}$ where $h$ is called the height of the packing.
The optimal (i.e. minimal) height
\kvn{$h$ for which it is possible to pack $\mathcal I$ in $\cont{A,B,h}$ is denoted by $\opts{\mathcal{I}, A, B}$.}

\begin{figure}[h]
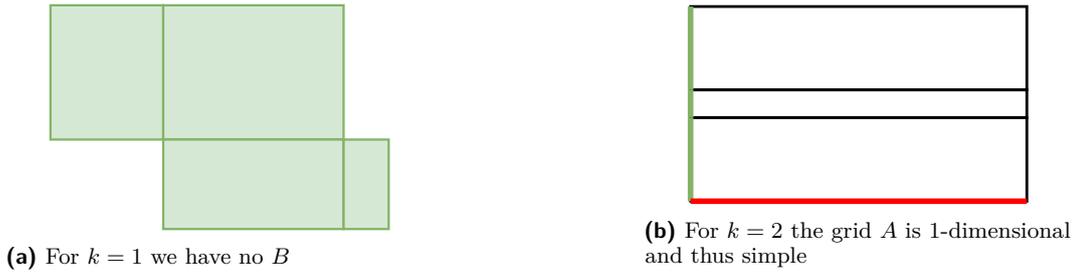

  \centering
  \begin{subfigure}{0.4\textwidth}
    \centering
    \includesvg[width=0.8\linewidth]{img/strip_basis_2d.svg}
    \caption{For $k=1$ we have no $B$}
    \label{fig:strip packing:2d base:1}
  \end{subfigure}
  \hfill
  \begin{subfigure}{0.4\textwidth}
    \centering
    \includesvg[width=0.8\linewidth]{img/strip_basis_2d2.svg}
    \caption{For $k=2$ the grid $A$ is $1$-dimensional and thus simple}
    \label{fig:strip packing:2d base:2}
  \end{subfigure}
  \caption{$2$-dimensional bases of a $3$-dimensional strip with $A$ shown in green (the short sides) and $B$ shown in red (the long sides)}
  \label{fig:strip packing:2d base}
\end{figure}

\begin{figure}[h]
  \centering
  \includesvg[width=0.5 \linewidth]{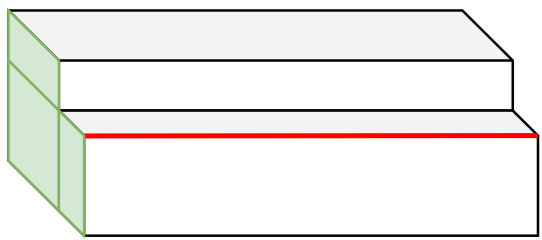}
  \caption{$3$-dimensional base \mar{with $k=2$} of a $4$-dimensional strip with $A$ shown in green (the short sides) and $B$ shown in red (the long sides)}
  \label{fig:strip packing:base}
\end{figure}

We only consider instances where a solution exists. Since we assume that $\vol[d-k]{A}=1$,
we have that $\hat s\le 1$ as a larger item would not fit into the strip.
Note that for each $i\in[k-1]$, $b_i \geq N\alpha^{d-k} \hat{s} \geq N(\alpha \hat{s})^{d-k}\ge\sum_{a\in A}\vol[d-k]{a} = \vol[d-k]{A} = 1$.
In this section, we describe an algorithm based on Harren's \cite{harren-journal} multidimensional generalization
of the algorithm by Kenyon and R{\'{e}}mila \cite{KenyonR00} for $2$-dimensional strip packing.

The differences between our algorithm and Harren's algorithm\cite{harren-journal} can be summed up as follows:
Due to the complex nature of $A$, we have multiple hypercuboids, each having a bounded aspect ratio.
This changes the analysis of the algorithm. It also changes the details like packing the medium and small items on the top,
because we have to split one strip into multiple smaller strips
(having hypercuboidal bases) and distribute the items among them. The addition of very long
sides in form of $B$, like in \cref{fig:strip packing:2d base:2}, breaks the
bounded aspect ratio property in a more crucial way. Harren's
algorithm first packs the large items on the base and then extends
this packing along the height of the strip. This relies on the fact that,
due to bounded aspect ratio, not many large items can be packed on the base,
and thus we have to consider only a constant number of so-called
\emph{configurations}. If we have additional long sides in the base in form of $B$,
we don't have such a bound on the number of large items
that can fit in the base. For example, in \cref{fig:strip packing:2d base:2},
an item can be large with respect to $A$ (shown in \mar{green}), but the number of
such items that can fit on the base $A\times B$ can't be bounded by a constant
since the side shown in \mar{red} can be arbitrarily long.

We work around this problem as follows: First, we consider $(d-k)$ dimensional packings,
i.e., \emph{configurations} of the items, on $A$.
Then we not only extend these configurations along the height of the strip but
also along each of the $(k-1)$ dimensions of $B$.
For example, in \cref{fig:strip packing:2d base:2}, we first create one-dimensional
packings on $A$ and then extend each of these packings
both in the rightward direction and along the height of the strip.

\begin{theorem}
\label{knapsack:structure:strip packing}
Let $\mathcal{I}$, $A$ and $B$ be the input for the $d$-dimensional strip packing problem defined above
and let $\epsilon > 0$ be some additional accuracy parameter.
Then there is an algorithm which packs all items of $\mathcal{I}$ into the region  $\cont{A, B + \hat{s}, h}$
where $B + \hat{s}$ is the hypercuboid $B$ after increasing each side length by $\hat{s}$
and the height $h$ is given by
\begin{align*}
h \leq \paren{1 + \mathcal{O}\paren{\epsilon}}\opts{\mathcal{I},A,B} + \mathcal{O}\paren{\hat{s}}
\end{align*}
\kvn{The constant omitted in the expression $\mathcal O(\eps)$ is a function of $k$. The constant
omitted in the expression $\mathcal O(\hat s)$ is a function of $d,k,\alpha,N,\eps$.}
\end{theorem}

Due to space limitations, the detailed proof of \cref{knapsack:structure:strip packing} is given in \cref{sec:strip packing}.
Our algorithm first classifies the input items into large, medium, and small items (See \Cref{sec:strip packing:beginning} for the classification). 
The side lengths of large items are rounded using linear grouping \cite{binpacking-aptas} such that there is only a constant number of different sizes.
When we try to pack a subset of large items by arranging them in the first $(d-k)$ dimensions while giving them the same position in the last $k$ dimensions, we only have to consider a constant number of those subsets.
This is because the volume of the strip in the first $(d-k)$ dimensions is small and thus only a small subset of large items is packable this way and items of the same size can be considered to be identical.
For each configuration, we take a packing, that only uses the first $(d-k)$ dimensions, and extend it to use also the last $k$ dimensions by placing multiple items of the same size next to each other in those long dimensions.
This creates an \Nbox{} from each item in each configuration.
An LP is used to determine how large a configuration is extended in those $k$ dimensions.
The packings for the different configurations are treated as layers that are stacked on top of each other. See \cref{sec:strip packing:large} for details of packing of large items.

After the large items are packed, we use the gaps between the large items to create \Vboxes{} to hold some of the small items.
See \cref{sec:strip packing:gaps} for details of packing small items into the gaps.
The remaining small items and medium items are placed in additional \Vboxes{} on top of this packing.
Again this needed several technical adaptations, as unlike \cite{harren-journal}, we have multiple different bases and we need to distribute the remaining items in a balanced way such that the total height is minimized.
See \cref{sec:strip packing:top} for the details.

\kvn{As we mentioned earlier, we use this theorem
to repack a $k$-elongated collection $Q$ ($k\in[d-1]$). So, we use $\alpha_{k+1}$ as $\alpha$ since we know that
for every cell in $\mathcal Q$, the first $k$ dimensions have length at least
$\alpha_k \hat{s}(\mathcal Q)\ge N \alpha_{k+1}^{d-k} \hat{s}(\mathcal Q)$ (see \cref{knapsack:structure:strip packing:configs}
and definition of $\alpha_k$)
and each of the last $(d-k)$ dimensions have
length at most $\alpha_{k+1}\hat{s}(\mathcal Q)$.}
Now let us note a few useful remarks \kvn{about the above theorem applied to our collection $Q$}.
\begin{remark}\label{knapsack:structure:strip packing:exact}
A more exact bound for the height of the packing obtained in \cref{knapsack:structure:strip packing} is
$$
\left(1 + \left(2 + \max\left(3,2^{k-1}\right)\right)\epsilon\right)\opts{Q,A,B}
+ 2\left(\Cname{configs}\left(d, k, \alpha_{k+1}, N, \epsilon\right)+3\right)\hat{s}\paren{\mathcal Q}
$$
Here $\Cname{configs}$ is only dependent on its parameters.
(Proven by \cref{strip packing:case 1} and \cref{strip packing:case 2} in \cref{sec:strip packing:analysis}.)
\end{remark}

\begin{remark}\label{knapsack:structure:strip packing:boxes}
The algorithm of \cref{knapsack:structure:strip packing} packs all items in \Nboxes{} and \Vboxes{}.
The number of \Nboxes{} is bounded by
$
\Cname{configs}\left(d,k,\alpha_{k+1},N,\epsilon\right)/C_{\rho}\left(d,k,\alpha_{k+1},N,\epsilon\right)^{d-k}
$
and the number of \Vboxes{} by
$
\Cname{configs}\left(d,k,\alpha_{k+1},N,\epsilon\right)N2^{d-k} / \Cname{\rho}\left(d,k,\alpha_{k+1},N,\epsilon\right)^{\left(d-k\right)^2} + N
$.
Here $\Cname{configs}$ and $\Cname{\rho}$ are only dependent on their parameters.
(Proven by \cref{lem:knapsack:structure:strip packing:boxes} in  \cref{sec:strip packing:analysis}.)
\end{remark}
\arir{This does not follow easily from appendix. So add an explicit proof.}
\begin{remark}\label{knapsack:structure:strip packing:configs}
The value $\Cname{configs}\left(d, k, \alpha_{k+1}, N, \epsilon\right)$ from \cref{knapsack:structure:strip packing:exact,knapsack:structure:strip packing:boxes} has a lower bound of $\Cname{configs}\left(d, k, \alpha_{k+1}, N, \epsilon\right)>N\alpha_{k+1}^{d-k}\geq \alpha_{k+1}$.
(Proven by \cref{lem:cconfig} in \cref{sec:strip packing:analysis}.)
\end{remark}


\subsubsection{Repacking a Collection}\label{sec:knapsack:structure:repacking}
Depending on the type of the collection,
we can simplify the packing of the items assigned to each collection.
Let $\mathcal C$ be a $k$-elongated collection and let $\mathcal{Q}$ be the
set of items assigned to \kvn{$\mathcal C$}.
Let $\hat{s} := \hat{s}\paren{\mathcal{Q}}$.
We consider the cases $k=d, 1\le k<d$, and $k=0$ separately.

\ari{First of all, if $\mathcal C$ is a $d$-elongated collection, then it consists of only a single cell as it has no bad facets and hence can not be merged with any other cell. 
Thus we can repack (after removing only a small profit subset) it efficiently using NFDH
because it is a hypercuboidal region that is large in all dimensions.}

\begin{lemma}\label{knapsack:structure:d-elongated}
If $\mathcal C$ is a $d$-elongated collection, then it can be transformed into a single 
\Vbox{} with a loss of profit at most $6d\epsilon\prof{\mathcal{Q}}$.
\end{lemma}

For a $k$-elongated collection with $1 \leq k < d$, we can apply the strip packing algorithm of \cref{knapsack:structure:strip packing}
\kvn{as shown in the next lemma}.

\begin{lemma}\label{knapsack:structure:k-elongated}
A $k$-elongated collection with $1 \leq k < d$ can be transformed into a constant number of \Vboxes{} and \Nboxes{} by losing profit at most
$\paren{k + 7 + \max\paren{6,2^k}}\epsilon \prof{\mathcal{Q}}$ \kvn{where $\mathcal Q$ is the set of items assigned to it}.
\end{lemma}
\begin{proof}
We begin by shrinking the collection to compensate for the enlargement that the strip packing algorithm will cause.
Define \kvn{one of the long dimensions} of the collection to be \kvn{its height and let $h$ be the length in this dimension}.
Let $c := 4+\max\paren{3,2^{k-1}}$.
\kvn{We assumed $\epsilon < 1/2^{(d+2)} < 1/(2c)$.}
Divide the strip into $\floor{1/\paren{c\epsilon}}$ slices along the height.
After that each slice has a height of at least
\begin{align*}
c\epsilon h
&= \paren{4+\max\paren{3,2^{k-1}}}\epsilon h\\
&\geq \paren{2+\max\paren{3,2^{k-1}}}\epsilon h
+ 2 \epsilon \alpha_k\hat{s}
\tag{since $h \geq \alpha_k\hat{s}$}\\
&{{=} \paren{2+\max\paren{3,2^{k-1}}}\epsilon h
+ 4\paren{\Cname{configs}\paren{d,k,\alpha_{k+1},N,\epsilon}+3}\hat{s}}\tag{{by definition of $\alpha_k$}}\\
&\geq \paren{2+\max\paren{3,2^{k-1}}}\epsilon h
+ 2\paren{\Cname{configs}\paren{d,k,\alpha_{k+1},N,\epsilon}+3}\hat{s}
+ 6 \hat{s}.
\end{align*}
By clearing the slice with the lowest profit \kvn{of items completely contained in it,} we leave a gap with a height of
$\paren{2+\max\paren{3,2^{k-1}}}\epsilon h
+ 2\paren{\Cname{configs}\paren{d,k,\alpha_{k+1},N,\epsilon}+3}\hat{s}
+ 4 \hat{s}$.
This causes a loss of profit at most
$
\frac{1}{\floor{1/\paren{c\epsilon}}}\prof{\mathcal{Q}}
\leq \frac{1}{1/\paren{c\epsilon} - 1}\prof{\mathcal{Q}}
\leq \frac{1}{1/\paren{c\epsilon} - 1/\paren{2c\epsilon}}\prof{\mathcal{Q}}
= 2c\epsilon\prof{\mathcal{Q}}.
$
The items intersecting the bottom or top facet of the collection fit in a gap of height $2\hat{s}$.
We reduce the height of the collection to
\[h' := h - \paren{2+\max\paren{3,2^{k-1}}}\epsilon h
- 2\paren{\Cname{configs}\paren{d,k,\alpha_{k+1},N,\epsilon}+3}\hat{s}\]
by shifting the items intersecting the (removed) region at the top into a gap of height
$\paren{2+\max\paren{3,2^{k-1}}}\epsilon h
+ 2\paren{\Cname{configs}\paren{d,k,\alpha_{k+1},N,\epsilon}+3}\hat{s} + \hat{s}$.

For each other large dimension we create $\ceil{1/\epsilon}$ slices and clear the slice with the lowest profit,
losing at most $\epsilon \prof{\mathcal{Q}}$
and create a gap of width at least
$\frac{1}{\ceil{1/\epsilon}}\alpha_k\hat{s}-2\hat s
\geq \frac{1}{1/\epsilon + 1}\alpha_k\hat{s}-2\hat s
\geq \frac{\epsilon}{2}\alpha_k\hat{s}-2\hat s
\geq 4 \hat{s}.$
Now we can shift the items intersecting the orthogonal facets into that gap
and reduce the length of the collection in this direction by $\hat{s}$.
After this shifting process, we will be left with a subset of items $\mathcal{Q}' \subseteq \mathcal{Q}$.

To be able to use the strip packing algorithm from \cref{knapsack:structure:strip packing}, the base of our strip has to be in a special representation.
\kvn{For this, }we create a $\paren{d-k}$-dimensional grid $A$ by projecting the cells of our collection to their $\paren{d-k}$ short dimensions.
We define $B$ as a tuple holding the lengths in the $\paren{k-1}$ long dimensions of the collection that are not the height.
Now our base can be represented as the cartesian product $A \times B$.
As a last preparation step, we scale the whole collection and the items in it by $f := 1/\vol[d-k]{A}^{1/\paren{d-k}}$
\kvn{in each dimension} to normalize the volume of $A$ \kvn{to $1$}.
Note that each value in the tuple $B$ is at least $\alpha_k \hat{s}f \geq N\alpha_{k+1}^{d-k}\hat{s}f$ by the definition of $\alpha_k$ and \cref{knapsack:structure:strip packing:configs}.
Thus, we can use the strip packing algorithm \cref{knapsack:structure:strip packing} using the bound on the height by \cref{knapsack:structure:strip packing:exact} to compute a packing of height at most
\begin{align*}
& \left(1 + \left(2 + \max\left(3,2^{k-1}\right)\right)\epsilon\right)\opts{Q',A,B}
+ 2\left(\Cname{configs}\left(d, k, \alpha_{k+1}, N, \epsilon\right)+3\right)\hat{s}f\\
&\leq \left(1 + \left(2 + \max\left(3,2^{k-1}\right)\right)\epsilon\right)h'f
+ 2\left(\Cname{configs}\left(d, k, \alpha_{k+1}, N, \epsilon\right)+3\right)\hat{s}f\\
&\leq h'f + \left(2 + \max\left(3,2^{k-1}\right)\right)\epsilon hf
+ 2\left(\Cname{configs}\left(d, k, \alpha_{k+1}, N, \epsilon\right)+3\right)\hat{s}f\\
&= hf.
\end{align*}
After scaling back, the packing fits into our collection.
\end{proof}
\kvn{
  \textbf{Repacking $0$-elongated Collections. }For 
  a $0$-elongated collection, we first distinguish the items into large, medium, and small items such that the profit of
  medium items is very small so that they can be discarded. We then further distinguish between the cases when (i) there is
  a large item of very small profit or (ii) every large item has a significant profit. In the first case, we remove the large item
  to make enough space to repack the small items. In the second case, we use the fact that the number of large items can only be
  $O(1)$. So, we partition the grid into smaller grids and solve these recursively. We finally prove that we only require $O(1)$ number of recursive
  steps that we require are only a constant in number.
}

\kvn{From now on, lets assume that $\mathcal C$ is a $0$-elongated collection.}
To partition $\mathcal Q$ into large, medium, and small items, we define the following values.
$\rho'_0 := 1, 
\rho'_{i+1} := \paren{\frac{\paren{\rho_i'}^d}{4dN\alpha_1^d}}^{d+1}$ and  
$\rho_i := \rho_i' \hat{s}$ for all $i \geq 0$.
Let $\eta \in \nat_+$ be minimal such that $
\prof{\Set{x\in \mathcal{Q} | \rho_{\eta} < s\left(x\right) < \rho_{\eta-1}}}
\leq \epsilon \prof{\mathcal{Q}}$.
\kvn{Note that $\eta\le1/\eps$.}
We now partition $\mathcal{Q}$ into sets
$\sitems := \Set{x\in \mathcal{Q} | s\left(x\right) \leq \rho_{\eta}}$,
$\mitems := \Set{x\in \mathcal{Q} | \rho_{\eta} < s\left(x\right) < \rho_{\eta-1}}$ { and } 
$\litems := \Set{x\in \mathcal{Q} | \rho_{\eta-1} \leq s\left(x\right)}$.
\kvn{
We call $\sitems$ (resp. $\mitems,\litems$) to be the set of small (resp. medium, large) items.
Note that since $\eta\ge1$ and $\rho_0=\hat s$, the set of large items can't be empty. This
simple observation will be useful later.
}

\ari{If there is a large item with a small profit, then after discarding that large item, the entire empty space can be divided into \kvn{a constant number of} \Vboxes{}. It can then be shown that these \Vboxes{} have enough volume to pack all the small items.}

\begin{lemma}\label{knapsack:structure:0-elongated:large with low profit}
\kvn{
Suppose $\mathcal C$ is a $0$-elongated collection.
If an item in $\litems$ has profit at most $\epsilon \prof{\mathcal{Q}}$, then the collection $\mathcal C$}
can be transformed into a \kvn{packing of profit at least $\paren{1-2\epsilon} \prof{\mathcal{Q}}$}
containing a constant number of large items
and \kvn{a constant number of} \Vboxes{}.
\end{lemma}

If all the large items have a good profit \kvn{(i.e., a profit of more than $\eps \prof{\mathcal Q}$)}, we cannot discard any of them to create space.
Therefore, it is hard to repack the items directly so instead we will take a recursive approach using a shifting argumentation.

\begin{lemma}\label{knapsack:structure:0-elongated:large with high profit}
\kvn{Suppose $\mathcal C$ is a $0$-elongated collection.}
If every item in $\litems$ has a profit of at least $\epsilon \prof{\mathcal{Q}}$, then $\mathcal C$
can be transformed into a constant number of large items, \Nboxes{} and \Vboxes{} while losing profit at most
$\max\paren{6d + 1, d + 13, d + 7 + 2^{d-1}}\epsilon\prof{\mathcal Q}$.
\end{lemma}

\begin{proof}
In this case, we have at most $1/\epsilon$ large items.
We can split the area of the collection, which is not covered by the large items into a grid of at most $2/\epsilon$ additional layers in each dimension
and recursively start from \cref{sec:knapsack:structure:classification} with the classification of cells and merging into collections.
\kvn{
Each of these collections is a $k$-elongated collection ($k>0$) or a $0$-elongated collection. The former collections
can be repacked using \cref{knapsack:structure:d-elongated} or \cref{knapsack:structure:k-elongated}.
Let's look at the $0$-elongated collections. Note that the total profit of these collections is at most $(1-\eps)\prof{\mathcal Q}$
since we recurse only if there is a large item in $\mathcal Q$ of profit at least $\eps\prof{\mathcal Q}$.
For a $0$-elongated collection in this recursive step, we
either use \cref{knapsack:structure:0-elongated:large with low profit} if there
exists a large item of small profit or, we continue the recursion.
}
Note, that if we reach a depth of $\ceil{\log_{1-\epsilon}\paren{\epsilon}}$ in this recursion,
then we can just discard the remaining
items and stop.
The reason for this is that we packed in each recursive step at least one large item with a non-negligible profit,
\mar{so at this maximal recursion depth, all} the items \mar{over all collections} have a profit at most
$\paren{1-\epsilon}^{\ceil{\log_{1-\epsilon}\paren{\epsilon}}}\prof{\mathcal{Q}} \leq \epsilon \prof{\mathcal{Q}}$.

During this process, the initial region was split into different collections,
where each of those was assigned a partition of the items $\mathcal{Q}' \subseteq \mathcal{Q}$.
At this point, no profit was \mar{lost except for the items that are lost at the maximal recursion depth with profit at most $\epsilon \prof{\mathcal{Q}}$}.
When repacking each collection by transforming them into large items, \Nboxes{} and \Vboxes{},
\mar{we lose again some profit, which depends on the method we used for repacking.
We lose at most $6d\epsilon\prof{\mathcal{Q}'}$ through \cref{knapsack:structure:d-elongated}, at most $\paren{d + 6 + \max\paren{6, 2^{d-1}}}\epsilon\prof{\mathcal{Q}'}$ through \cref{knapsack:structure:k-elongated}, at most $2\epsilon\prof{\mathcal{Q}'}$ through \cref{knapsack:structure:0-elongated:large with low profit}.}
As those subsets for the collections are distinct the loss in this repacking step is bounded by $\max\paren{6d, d + 12, d + 6 + 2^{d-1},2}\epsilon\prof{\mathcal{Q}}$.
\mar{\kvn{Noting that $2<\max\paren{6d, d + 12, d + 6 + 2^{d-1}}$ and adding the factor of $\eps\prof{\mathcal Q}$ that we may lose during the recursion} proves the claim.}
\end{proof}

\arir{commented new subsection}
The last thing to do is to prove, that the number of \Nboxes{} and \Vboxes{} that are created is constant.
As those boxes are created out of the cells of the grid in \cref{sec:knapsack:structure:classification}, we start by bounding the size of this grid during the recursive algorithm. \ari{The following lemma follows from the fact that in each of the recursive step (having depth at most $\lceil \log_{1-\eps}\eps \rceil$) the number of layers increases by at most $2/\eps$ in each of the dimensions.} 
\begin{lemma}\label{knapsack:structure:constants:N}
\kvn{Consider any grid $\mathcal G$ obtained during the recursive algorithm starting from the knapsack with an optimal packing as a grid with a single cell. Then the number of layers in each dimension in $\mathcal G$ is upper bounded by
$\Cname{layer}\paren{\epsilon} := 2\ceil{\log_{1-\epsilon}\paren{\epsilon}}/\epsilon\kvn{+1}$.}
\end{lemma}

\kvn{
\begin{remark}
\label{num-cells-grid}
Due to the above lemma, the bound on the number of layers also gives an upper bound $\Cname{N}\paren{d, \epsilon} := \paren{\Cname{layer}\paren{\epsilon}}^d$ on the number of cells in any grid that we consider in \cref{sec:knapsack:structure:classification}.
\end{remark}
Using these results, we obtain the following lemma which bounds the number of \Nboxes{} and \Vboxes{}.
}

\begin{lemma}\label{knapsack:structure:constants:boxes}
\ari{The numbers of large items, and the total number of \Nboxes{} and \Vboxes{} generated by the transformation of an optimal packing can be bounded by $\Cname{large}$ and $\Cname{boxes}$, respectively.
$\Cname{large}$ and $\Cname{boxes}$ are constants  that depend  only on} $d, \epsilon$.
\end{lemma}
\mar{Now we can prove \cref{knapsack:structure}.
\begin{proof}[Proof of \cref{knapsack:structure}]
\ari{Consider an optimal packing and interpret the initial knapsack as a cell of a grid with one layer in each dimension.}
Then we start with the procedure explained in \cref{sec:knapsack:structure:classification}
to classify the unit hypercube either as $d$-elongated or $0$-elongated.
By \cref{knapsack:structure:d-elongated,knapsack:structure:0-elongated:large with low profit,knapsack:structure:0-elongated:large with high profit},
 we can simplify the structure of the packing and lose a profit of at most
$\max(6d + 1, d + 13, d + 8 + 2^{d-1})\epsilon \prof{\mathcal{I}}
\leq 2^{d+2}\epsilon \prof{\mathcal{I}}$
\kvnr{Commented some inequalities which I felt take up space.}
as we have $d \geq 2$.
\cref{knapsack:structure:constants:boxes} bounds \kvn{the number of \Vboxes{} and \Nboxes{}
and the large items packed outside them}.
\end{proof}}

\subsection{Algorithm}\label{sec:knapsack:algorithm}
Using the results of \cref{sec:knapsack:structure}, we can construct a PTAS for \dsqks{$d$}.

\begin{theorem}\label{knapsack:structure:algorithm}
Let $d \geq 2$ and $\epsilon > 0$.
There is an algorithm which returns for each instance of the $d$-dimensional hypercube knapsack packing problem \mar{given by a set of items }
a packing with profit at least $\paren{1-\epsilon}\optk{\mathcal{I}}$
with a running time which is polynomial in $\abs{\mathcal{I}}$.
\end{theorem}
\begin{proof}
Using \cref{knapsack:structure} with an accuracy of $2^{-d-3}\epsilon$
we know that there is a packing with a simple structure and profit at least $\paren{1-\epsilon/2}\optk{\mathcal{I}}$.
\mar{Let $\mathcal B$ be the set of \Nboxes{} and \Vboxes{} in this
structure, $\mathcal{L} \subseteq \mathcal{I}$ the set of items packed outside of those boxes and $\mathcal{S} \subseteq \mathcal{I} \backslash \mathcal{L}$ the set of items packed inside of them.}

As the number of items in $\mathcal L$ is constant,
there are at most $\abs{\mathcal{I}}^{\Cname{large}\paren{d, \epsilon}}$ choices for this subset $\mathcal{L} \subseteq \mathcal{I}$.
There are at most $\Cname{boxes}\paren{d, \epsilon} + 1$ choices for the number of boxes in $\mathcal B$
and for each box there are at most:
\begin{itemize}
\item $2$ choices whether it is a \Vbox{} or an \Nbox{},
\item $\abs{\mathcal{I}}$ choices for the size parameter $\hat{s}$ of the box,
\item $\abs{\mathcal{I}}^d$ choices for the side lengths of the box.
\end{itemize}
\mar{All of these choices are polynomial in the number of items $\abs{\mathcal{I}}$.}
Thus, by iterating over all possible choices,
we can assume at this point that we know the set of items $\mathcal{L}$
and the set $\mathcal B$ of \Vboxes{} and \Nboxes{} of that nearly optimal solution.
As both $\abs{\mathcal{L}}$ and $\abs{\mathcal B}$ are bounded by a constant,
we can find a packing of those items and boxes in the unit hypercube in constant time.
The last step left to do is to pack items from $\mathcal{I} \backslash \mathcal{L}$ in the boxes with a nearly optimal profit.
\ari{We can do this by solving a special variant of the Generalized Assignment Problem (GAP) \cite{shmoys1993approximation}.
In GAP, we are given a set of one-dimensional knapsacks with a capacity each and a set of items that may have a possibly different size and profit for each knapsack.
The goal in this problem is to find a feasible packing with maximal profit.
We will consider our $O(1)$ number of \Vboxes~and \Nboxes~to be knapsacks. 
Consider a \Vbox{} $B_V$ with size parameter $\hat s$. We set its capacity to be $\vol{B}-\hat s\left(\surf{B}/2\right)$.
For any item $i$, we set its size with respect to $B_V$ as $\vol{i}$ if $s(i)\le\hat s$ and $\infty$ otherwise.
On the other hand, consider an \Nbox{} $B_N$ with size parameter $\hat s$ and $\{n_i\}_{i\in{d}}$ denoting the number of cells in each dimension.
We set its capacity to be $n_1n_2\dots n_d$.
For any item $i$, we set its size with respect to $B_N$ as $1$ if $s(i)\le\hat s$ and $\infty$ otherwise.
The profit of any item $i$ with respect to any box is just set as $\prof{i}$.
This boils down to a variant of GAP with $O(1)$ number of knapsacks,
which admits a PTAS \cite{GalvezGHI0W17}.}
Thus by solving this instance, we get a subset $\mathcal{S}' \subseteq \mathcal{I} \backslash \mathcal{L}$ with profit at least $\prof{\mathcal{S}'}\geq\paren{1-\epsilon/2}\prof{\mathcal{S}}$.
Using \cref{n-box} and \cref{v-box} we can ensure to pack those items inside the boxes.
The total profit packed is
$\prof{\mathcal{S}'} + \prof{\mathcal L}
\geq \paren{1-\frac{\epsilon}{2}}\prof{\mathcal{S}} + \prof{\mathcal L}
> \paren{1-\frac{\epsilon}{2}}\paren{\prof{\mathcal{S}} + \prof{\mathcal L}}
\geq \paren{1-\frac{\epsilon}{2}}^2\optk{\mathcal{I}}
\mar{\geq \paren{1-\epsilon}\optk{\mathcal{I}}}.
$
\end{proof}

\section{Conclusion}
We have designed a PTAS for a variant of the knapsack problem (\dsqks{$d$}),
where the items are $d$-dimensional hypercubes with arbitrary profits.
En route, we have also developed a near-optimal algorithm for a variant of the $d$-dimensional hypercube strip packing problem, where the base need not be hypercuboidal, but we are
allowed to extend some of the sides of the base (the long dimensions) by a small amount.
Extending the techniques of \cite{HeydrichWiese2017}, we believe that it might be possible to design
an EPTAS for \dsqks{$d$}.

The knapsack problem for squares, cubes, and hypercubes is thus almost settled.
Whether there exists a PTAS for the knapsack variant where items are rectangles
is still open. It is also interesting to improve the current best approximation
ratios ($(5+\eps)$ with rotations and $(7+\eps)$ without rotations
\cite{diedrich2008approximation}) for the knapsack
problem where items are cuboids (3D) because of its practical relevance.
Another interesting task of theoretical importance would be to improve the
current best approximation ratio (which is $(3^d+\eps)$ due to \cite{Sharma21fsttcs}) for the knapsack
problem where items are $d$-dimensional hypercuboids.

\bibliography{literature}
\appendix
\section{Next Fit Decreasing Height}
\label{app:nfdh}
Next Fit Decreasing Height (NFDH) is a greedy algorithm introduced in \cite{coffman1980performance}
for packing a given set of rectangles in a bigger rectangular region $\mathcal R$.
Informally, it works as follows. First, it sorts the given
set of rectangles in decreasing order of heights (the vertical dimension). Then, it
starts off by packing the rectangles on the base of $\mathcal R$ side-by-side as much as
possible. Then, this level is closed, i.e., the base of $\mathcal R$ is shifted vertically to the
top of the first rectangle packed. In the next step, the remaining rectangles are packed on the new base
to whatever extent is possible. This process continues.
See \cref{fig:2-nfdh} for an illustration.
\begin{figure}[h]
\centering
\includesvg[width=0.3\textwidth]{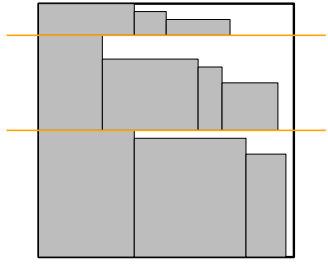}
\caption{A set of ten rectangles packed into a bigger square using NFDH}
\label{fig:2-nfdh}
\end{figure}

One can easily extend the NFDH algorithm to pack hypercuboids into a bigger hypercuboidal
region. Let $k\in\nat_{\ge2}$.
For any $k$-dimensional hypercuboid $x$, let $x^{(i)}$
denote the $i$-dimensional hypercuboid obtained by considering only the first $i$ dimensions of $x$.
We can extend this notation to sets: Let $J$ be a set of $k$-dimensional hypercuboids. Then
$J^{(i)}=\{x^{(i)}|x\in J\}$.

We define the NFDH algorithm in $k$ dimensions ($k$-NFDH) recursively as follows.
If $k=2$, then the algorithm is simply NFDH as described above. Suppose $k>2$.
Let $S$ be a set of $k$-dimensional hypercuboids to be packed into a $k$-dimensional hypercuboidal region $\mathcal R$. We first sort the rectangles in decreasing
order of their lengths in the $k\Th$ dimension. Then we pick the largest possible
prefix $P$ of $S$ such that $P^{(k-1)}$ can be completely packed on the base of $\mathcal R$ using $(k-1)$-NFDH.
\marr{Do we want to mention after the pseudo-code that in the case of hypercubes we don't have reorder the items in each step and thus can hand the full set of items to $(k-1)$-NFDH instead of just a prefix that we have to calculate?\kvn{Mentioned. Check.}}
We pack $P$ on the base of $\mathcal R$, and shift the base to the top of the tallest (longest in
the $k\Th$ dimension) item in $P$. We repeat this process with the set $S\backslash P$.
The pseudo-code can be found in \cref{alg:k-nfdh}.
\begin{algorithm}
\caption{
	$k$-NFDH$\left(S,\prod_{i=1}^{k}[a_i,b_i]\right)$: Pack a set $S$ of $k$-dimensional hypercuboids in a bigger
	$d$ dimensional hypercuboidal region $\mathcal R$ given by $\prod_{i=1}^{k}[a_i,b_i]$ where each $a_i,b_i\in\real$.
}
\begin{algorithmic}[1]
\If{$k==2$}
\State Use NFDH to pack $S$.
	\State\Return
\EndIf
\State Sort the items in $S$ in decreasing order of their lengths in the last ($k\Th$) dimension.
\State Pick the largest prefix $P$ of $S$ such that $P^{(k-1)}$ can be packed in the region $\prod_{i=1}^{k-1}[a_i,b_i]$.
\If{$P==\phi$}
	\State\Return
\EndIf
\State $\ell\leftarrow$ The length in the last dimension of the first item in $P$.
\If{$a_k+\ell>b_k$}
\State\Return\Comment{Not enough space in $\mathcal R$}
\EndIf
\State Pack $P$ in $\prod_{i=1}^{k-1}[a_i,b_i]\times[a_k,a_k+\ell]$ using $(k-1)$-NFDH$\left(P^{(k-1)},\prod_{i=1}^{k-1}[a_i,b_i]\right)$.
\State Run $k$-NFDH$\left(S\backslash P,\prod_{i=1}^{k-1}[a_i,b_i]\times [a_k+\ell,b_k]\right)$.
\end{algorithmic}
\label{alg:k-nfdh}
\end{algorithm}

We remark that when NFDH is applied to hypercubes in particular, we don't need to find the largest prefix $P$ of $S$ such that
$P^{(k-1)}$ can be packed using $(k-1)$ dimensional NFDH. Rather, we can pass the item set in the same order to the recursive step because the ordering of the hypercubes in any dimension would be the same.
An illustration of $3$-NFDH for cubes can be found in \cref{fig:3-nfdh}.
\begin{figure}
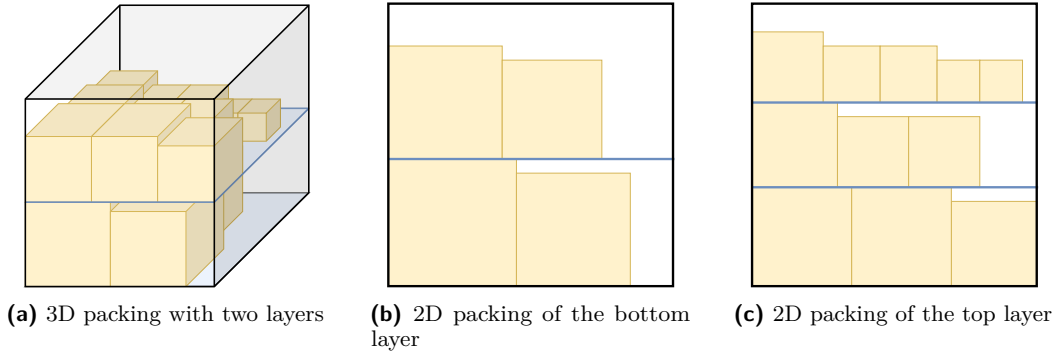

  \centering
  \begin{subfigure}[t]{0.3\textwidth}
    \centering
    \includesvg[width=0.9\linewidth]{img/3-nfdh} 
    \caption{3D packing with two layers}
    \label{fig:3-nfdh:1}
  \end{subfigure}
  \hspace{1em}
  \begin{subfigure}[t]{0.3\textwidth}
    \centering
    \includesvg[width=0.9\linewidth]{img/3-nfdh_layer_1}
    \caption{2D packing of the bottom layer}
    \label{fig:3-nfdh:2}
  \end{subfigure}
  \hspace{1em}
  \begin{subfigure}[t]{0.3\textwidth}
    \centering
    \includesvg[width=0.9\linewidth]{img/3-nfdh_layer_2}
    \caption{2D packing of the top layer}
    \label{fig:3-nfdh:3}
  \end{subfigure}
  \caption{A set of 15 cubes packed into a bigger cube using NFDH}
  \label{fig:3-nfdh}
\end{figure}

\section{Omitted Proofs}
\label{sec:omitpf}
\subsection{Proof of \cref{knapsack:structure:d-elongated}}
\begin{proof}
We assume $\epsilon \leq 1/2^d$.
A $d$-elongated collection consists of only a single cell $a$
as it has no \wide{} facets and hence can't be merged with any other cell.
Let $a'$ be the hypercuboid after rounding all the side lengths down
to the nearest multiple of $\hat{s}$
and let $a''$ be the hypercuboid obtained by multiplying each side length of $a$ by $\paren{1-3\epsilon/4}$.
We start by shifting all the items into $a''$ as follows.
We iterate over all dimensions.
Let $l$ be the side length of $a$ in the current dimension.
\arir{Subfigures have different sizes.}
We split $a$ along this dimension into $\floor{1/\paren{2\epsilon}}$ slices of the same length (\cref{fig:knapsack:structure:d-elongated:1})
and discard all the items completely contained inside the slice
with minimal profit (\cref{fig:knapsack:structure:d-elongated:2}).
Using $\epsilon \leq 1/3$, this creates a profit loss of at most
\begin{align*}
\frac{1}{\floor{1/\paren{2\epsilon}}}\prof{\mathcal{Q}}
&\leq \frac{1}{1/\paren{2\epsilon} - 1}\prof{\mathcal{Q}}
\leq \frac{1}{1/\paren{2\epsilon} - 1/\paren{3\epsilon}}\prof{\mathcal{Q}}
= 6\epsilon\prof{\mathcal{Q}}
\end{align*}
and leaves an empty space with length at least $2\epsilon l - 2\hat{s} \geq 3\epsilon l/4 + 3\hat{s}$,
since $\epsilon l \geq \epsilon \alpha_d \hat{s}=2d\hat s \geq 4\hat{s}$.
The items intersecting the top region of $a$ with length $3\epsilon l/4$
can be packed into a gap of length $3\eps l/4 + 2 \hat{s}$
and the items intersecting the bottom facet of $a$ can be
packed into a gap of length $\hat{s}$ (\cref{fig:knapsack:structure:d-elongated:3}).
\kvnr{Can we remove the rounded down part
  on the right side? As we are dealing with only one dimension.}
\begin{figure}
  \centering
  \begin{subfigure}[t]{0.3\textwidth}
    \centering
    \includesvg[width=0.83\linewidth]{img/knapsack_shifting_1}
    \caption{Hypercuboid $a$ (blue and red) split into three slices. Hypercuboid $a''$ is marked blue}
    \label{fig:knapsack:structure:d-elongated:1}
  \end{subfigure}
  \hspace{1em}
  \begin{subfigure}[t]{0.3\textwidth}
    \centering
    \includesvg[width=0.7\linewidth]{img/knapsack_shifting_2}
    \caption{One slice with minimal profit is cleared}
    \label{fig:knapsack:structure:d-elongated:2}
  \end{subfigure}
  \hspace{1em}
  \begin{subfigure}[t]{0.3\textwidth}
    \centering
    \includesvg[width=0.7\linewidth]{img/knapsack_shifting_3}
    \caption{Slice filled with items from the top region and the bottom line}
    \label{fig:knapsack:structure:d-elongated:3}
  \end{subfigure}
  \caption{Shifting technique}
  \label{fig:knapsack:structure:d-elongated}
\end{figure}

After this shifting in $d$ dimensions, the remaining items have profit at least $\paren{1-6d\epsilon}\prof{Q}$
and a volume of at most
\kvnr{changed the below set of inequalities. Check once.}
\marr{changed the below set of inequalities again.}
\begin{align*}
\vol{a''} &= \paren{1-\frac{3}{4}\epsilon}^d\vol{a}\tag{by definition of $a''$}\\
&\leq \paren{1-\frac{d+1}{2d}\epsilon}^d\vol{a}\tag{since $d\ge 2$}\\
&\leq \paren{1-\frac{\epsilon}{2}}^d\paren{1-\frac{\epsilon}{2d}}^d\vol{a}\tag{as $(1-x)(1-y)\geq 1-x-y$ if $x,y\ge0$}\\
&\leq \paren{1-\frac{\epsilon}{2}}\paren{1-\frac{\epsilon}{2d}}^d\vol{a}\tag{since $1-\eps/2 \in (0,1)$}\\
&= \paren{1-\frac{d}{\alpha_d}}\paren{1-\frac{1}{\alpha_d}}^d\vol{a}\tag{since $\alpha_d=2d/\epsilon$}\\
&\leq \paren{1-\frac{d}{\alpha_d}}\vol{a'}\tag{since $(1-1/\alpha_d)l \leq l - \hat{s}$ if $l \geq \alpha_d\hat{s}$}\\
&\leq \paren{1-\frac{d}{2\paren{\alpha_d - 1}}}\vol{a'}\tag{since $\alpha_d=2d/\eps>2$}\\
&=\vol{a'} - \frac{\hat{s}}{2}\sum_{i=1}^d\frac{\vol{a'}}{\alpha_d \hat{s} - \hat{s}}\\
&\leq\vol{a'} - \hat{s}\frac{\surf{a'}}{2}\tag{since each side of $a'\ge l-\hat s\ge\alpha_d\hat s-\hat s$}
\end{align*}
Thus, \kvn{from \cref{v-box}}, we can repack those items into $a'$ using NFDH.
\end{proof}


\subsection{Proof of \cref{knapsack:structure:0-elongated:large with low profit}}
\begin{proof}
First, we will need the following simple result.
\begin{claim}
\label{cuboid-compress}
Let $B$ be a $d$-dimensional hypercuboid and let $B'$ be the hypercuboid obtained by decreasing the length of $B$ in
every dimension by a value of at most $s$. Then we have:
$\vol{B'}-s\frac{\surf{B'}}{2}\ge\vol{B}-s\surf{B}$.
\end{claim}
\begin{proof}
We will first show that $\vol{B}-\vol{B'}\le s\surf{B}/2$. Let $B$ have lengths $b_1,b_2,\dots,b_d$. Then
\begin{align*}
\vol{B}-\vol{B'}&\leq\prod_{i\in[d]}b_i - \prod_{i\in[d]}\paren{b_i-s}\\
&\le \prod_{i\in[d]}b_i - \paren{\prod_{i\in[d]}b_i - s\sum_{i\in[d]}\prod_{j\in[d],j\ne i}b_j}
\tag{inductive argument over $d$}\\
&=s\frac{\surf{B}}{2}
\end{align*}
Thus we have $\vol{B'} \geq \vol{B} - s\surf{B}/2$.
The claim then follows by using the fact that $\surf{B'}\le \surf{B}$.
\end{proof}

Now we start by discarding the medium items,
losing profit most $\epsilon \prof{\mathcal{Q}}$.
The volume of the collection is at most $N\paren{\alpha_1 \hat{s}}^d$
(as there are at most $N$ cells and each cell has length at most $\alpha_1\hat s$ in each dimension)
and each large item has a size at least $\rho'_{\eta-1} \hat{s}$.
Therefore, there are at most $N \alpha_1^d / \rho_{\eta-1}^{'d}$ large items.
Recall that the grid that we consider has at most $\Nname{layer}$ number of layers in each dimension.
Using \cref{grid:splitting}, we can split the empty space into a grid $C$ and bound the number of cells by
\begin{align}
\abs{C} \leq \paren{\Nname{layer} + \frac{2N\alpha_1^d}{\rho_{\eta-1}^{'d}}}^d
&= \paren{N^{1/d} + \frac{2N\alpha_1^d}{\rho_{\eta-1}^{'d}}}^d
\leq \paren{\frac{N\alpha_1^d}{\rho_{\eta-1}^{'d}} + \frac{2N\alpha_1^d}{\rho_{\eta-1}^{'d}}}^d
= \paren{\frac{3N\alpha_1^d}{\rho_{\eta-1}^{'d}}}^d
\label{eq:bound-num-cells}
\end{align}
\arir{The second "<=" is not clear.}
\kvnr{Added a line. Check}
\kvn{The second inequality above follows since $\alpha_1, N\ge 1$ and $\rho_{\eta-1}'\le1$}.
We now discard the large item with the least profit,
losing an additional profit of at most $\epsilon \prof{\mathcal{Q}}$,
to create another empty hypercuboidal space $B$ to pack small items in.
Let $\check{s} := \hat{s}\paren{\sitems}$ be the size of the largest small item.
We round the sizes of the cells in $C$ and the hypercuboid $B$ down to the closest multiple of $\check{s}$ and obtain a set of cells $C'$  and hypercuboid $B'$.
Now we can pack as many small items as possible into the rounded hypercuboids in $C'$ using the method for a \Vbox{} from \cref{v-box} (i.e., using NFDH).
We claim that the remaining small items $\sitems'$ can be packed into $B'$.
To show that this is possible, we have to show
that the volume of the items in $\sitems'$ is at most
the volume which we can pack in $B'$ using NFDH (\cref{v-box}).
We know that the volume of all of the small items is at most the volume of the cells in $C$ and
in each cell $a'\in C'$,
we pack a volume of at least $\vol{a'}-\check{s}\left(\surf{a'}/2\right) - \check{s}^d$
as the volume of any small item is at most $\check s^d$.
Thus, the volume of the items in $\sitems'$ is at most:
\begin{align}
\vol{\sitems'}
&\leq \sum_{a \in C}\vol{a}
- \sum_{a' \in C'}\paren{\vol{a'} - \check{s}\cdot\frac{\surf{a'}}{2} - \check{s}^d}\nonumber\\
&\leq \sum_{a \in C}\vol{a}
- \sum_{a \in C}\paren{\vol{a} - \check{s}\surf{a} - \check{s}^d}
\tag{by \cref{cuboid-compress}}\nonumber\\
&= \check{s}\sum_{a \in C}\paren{\surf{a} + \check{s}^{d-1}}\nonumber\\
&\leq \check{s}\abs{C}\paren{2d\paren{\alpha_1\hat{s}}^{d-1} + \check{s}^{d-1}}\nonumber\\
&\le \check{s}\abs{C}\hat{s}^{d-1}\paren{2d\alpha_1^{d-1} + 1}\tag{as $\check{s} < \hat{s}$}\nonumber\\
&= \check{s}\abs{C}\hat{s}^{d-1}\paren{4d\alpha_1^{d-1} + 1}
- \check{s}\abs{C}\hat{s}^{d-1}2d\alpha_1^{d-1}\nonumber\\
&\leq \check{s}\abs{C}\hat{s}^{d-1}\paren{4d\alpha_1^{d-1} + 1}
- \check{s}2d\hat{s}^{d-1}\tag{as $\abs C\ge 1$ and $\alpha_1\ge 1$}\nonumber\\
&< \check{s}\abs{C}\hat{s}^{d-1}8d\alpha_1^{d-1} \tag{as $4d\alpha_1^{d-1} > 1$}- \check{s}2d\hat{s}^{d-1}\nonumber\\
&< \check{s}\abs{C}\hat{s}^{d-1}4d\alpha_1^d \tag{as $\alpha_1 > \alpha_d > 2$}- \check{s}2d\hat{s}^{d-1}\nonumber\\
&\leq \paren{\frac{\rho_{\eta-1}^{'d}}{4dN\alpha_1^d}}^{\!\!d+1}\!\!\!\!
\abs{C}\hat{s}^d 4d\alpha_1^d \tag{\kvn{by bounding $\check{s}$ by $\check s\le\rho_{\eta} = \rho'_{\eta} \hat{s}$}}- \check{s}2d\hat{s}^{d-1}\nonumber\\
&\leq \paren{\frac{\rho_{\eta-1}^{'d}}{4dN\alpha_1^d}}^{d+1}\!
\paren{\frac{3N\alpha_1^d}{\rho_{\eta-1}^{'d}}}^d
\hat{s}^d 4d\alpha_1^d - \check{s}2d\hat{s}^{d-1}\tag{by \cref{eq:bound-num-cells}}\nonumber\\
&< \frac{\rho_{\eta-1}^{'d}}{4dN\alpha_1^d}\;\;
\hat{s}^d 4d\alpha_1^d\nonumber\\
&= \frac{1}{N} \rho_{\eta-1}^{'d}\hat{s}^d- \check{s}2d\hat{s}^{d-1}\nonumber\\
&\leq \rho_{\eta-1}^{'d}\hat{s}^d- \check{s}2d\hat{s}^{d-1}
\label{residual-small-1}
\end{align}
On the other hand,
\begin{align}
\vol{B'} - \check{s}\cdot\frac{\surf{B'}}{2}
&\geq \vol{B} - \check{s}\surf{B}\nonumber\tag{\kvn{by \cref{cuboid-compress}}}\\
&\geq \rho_{\eta-1}^d - \check{s}\surf{B}\tag{$B$ has the size of a large item}\nonumber\\
&\geq \rho_{\eta-1}^{'d}\hat{s}^d - \check{s}\surf{B}\tag{by definition of $\rho_{\eta-1}$}\nonumber\\
&\geq \rho_{\eta-1}^{'d}\hat{s}^d - \check{s}2d\hat{s}^{d-1}\label{residual-small-2}
\end{align}
The last inequality follows since the hypercuboid $B'$ is the
space created by a large item which has side length at most $\hat s$.
Combining \cref{residual-small-1,residual-small-2}, we obtain that
\begin{align*}
\vol{\sitems'}\le\vol{B'}-\check s\cdot\frac{\surf{{B'}}}{2}
\end{align*}
Thus, using \cref{v-box},
\kvn{the set $\mathcal S'$ can be packed in $B'$.}
\end{proof}
\subsection{Proof of \cref{knapsack:structure:constants:N}}
\begin{proof}
The algorithm starts with a single hypercube, i.e., the knapsack, which corresponds to a grid with $1 < 2/\epsilon+1$ layer in each dimension.
As mentioned in the proof of \cref{knapsack:structure:0-elongated:large with high profit}, the number of large items is
upper bounded by $1/\eps$ if we recurse.
Thus, in each recursive step, the number of layers is increased by at most $2/\epsilon$ in each dimension.
This implies that in any grid at a recursion depth of $i$, the number of layers
in each dimension is bounded by $2i/\epsilon+1$.
As in the proof of \cref{knapsack:structure:0-elongated:large with high profit},
\kvn{the maximal depth of the recursion is $\ceil{\log_{(1-\eps)}\eps}$. Hence the claim follows.}
\end{proof}

\subsection{Proof of \cref{knapsack:structure:constants:boxes}}
\begin{proof}
To start off, we have the knapsack which is just one cell. By \cref{num-cells-grid}, each cell can give rise to at most
$\Cname{N}\paren{d, \epsilon}$ number of cells. Since we recurse only until a depth of
$\ceil{\log_{1-\epsilon}\paren{\epsilon}}$, the total number of cells is upper bounded by
\begin{align*}
\Cname{cells}\paren{d, \epsilon} &:= \paren{\Cname{N}\paren{d, \epsilon}}^{\ceil{\log_{1-\epsilon}\paren{\epsilon}}}
\end{align*}
Let us have a look at the values that were used to classify cells and items as large or small, as they limit the number of \Nboxes{} and \Vboxes{}.
In \cref{sec:knapsack:structure:classification}, the value of $\alpha_d$ is only dependent on $\epsilon$ and $d$.
The remaining values $\alpha_k$ are defined inductively and each can be bounded using only $\epsilon$ and $d$, as the number of cells $N$ is also bounded by $\eps,d$ ($1\leq N \leq \Cname{N}\paren{d, \epsilon}$).
The value of $\rho'_0$ is defined as $1$.
\kvn{Each $\rho_i'$ is defined inductively and we can bound $\eta$ to be at most}
$\ceil{1/\epsilon}$. Thus there is a constant $\Cname{\rho}\paren{d, \epsilon} \leq \rho'_{\eta-1}$.

Now we can bound the number of large items in the transformed packing.
Only while repacking $0$-elongated collections (\cref{knapsack:structure:0-elongated:large with low profit}
and \cref{knapsack:structure:0-elongated:large with high profit}), we place items outside of \Nboxes{} and \Vboxes{}.
At the start of \cref{knapsack:structure:0-elongated:large with low profit}, we proved that in any $0$-elongated collection
with $N$ cells, the number of large items is at most $N\left(\alpha_1/\rho'_{\eta-1}\right)^d$.
At most $\Cname{cells}\paren{d, \epsilon}$ number of cells are generated in our algorithm.
So, the number of large items is upper bounded by
\begin{align*}
\Cname{large}\paren{d, \epsilon} &:= \Cname{cells}\paren{d, \epsilon}
\Cname{N}\paren{d, \epsilon}\frac{\alpha_1^d}{\Cname{\rho}\paren{d, \epsilon}^d}.
\end{align*}

To analyze the number of \Vboxes{} and \Nboxes{}, we will consider the number of boxes that are created from a single $k$-elongated collection for each $0\leq k\leq d$ and find an upper bound $\Cname{boxes}\paren{k, d, \epsilon}$ for that case.
The case of a $d$-elongated collection (\cref{knapsack:structure:d-elongated})
has the upper bound $\Cname{boxes}\paren{d, d, \epsilon} := 1$.
For a $0$-elongated collection we don't create boxes in the recursive case
(\cref{knapsack:structure:0-elongated:large with high profit})
and in the non-recursive case (\cref{knapsack:structure:0-elongated:large with low profit}), there are at most
$\Cname{boxes}\paren{0, d, \epsilon} := 1 + \paren{\Cname{layer}\paren{\epsilon} + 2\Cname{N}\paren{d, \epsilon}\alpha_1^d / \Cname{\rho}\paren{d, \epsilon}^d}^d$ boxes created.
For the remaining values of $k$, the number of boxes $\Cname{boxes}\paren{k, d, \epsilon}$ is
given by \cref{knapsack:structure:strip packing:boxes}
as the algorithm of \cref{knapsack:structure:strip packing} is used.
Thus, the total number of \Nboxes{} and \Vboxes{} over all collections is at most
\begin{align*}
\Cname{boxes}\paren{d, \epsilon} &:= \Cname{cells}\paren{d, \epsilon}
\max_{0 \leq k \leq d}\paren{\Cname{boxes}\paren{k, d, \epsilon}}.\qedhere
\end{align*}
\end{proof}


\section{Strip Packing with Resource Augmentation}\label{sec:strip packing}

\subsection{Classification of the Items}\label{sec:strip packing:beginning}

We assume that the accuracy parameter $\epsilon$ is at most $1/3$.
We partition the input set $\mathcal{I}$ of hypercubes into sets of
large $\mathcal{L} := \Set{x \in \mathcal{I} | \rho_{\eta} \leq s\paren{x}}$,
medium $\mathcal{M} := \Set{x \in \mathcal{I} | \rho_{\eta+1} < s\paren{x} < \rho_{\eta}}$ 
and small $\mathcal{S} := \Set{x \in \mathcal{I} | s\paren{x} \leq \rho_{\eta + 1}}$ items where
\begin{align*}
\rho_1 := \frac{\epsilon\hat{s}}{2\left(d-1\right)N\left(\alpha\hat{s}\right)^{d-k}}\textup{ and }
\rho_{i+1} := \rho_1 \frac{\rho_i^{\left(d-k\right)^2}}{2^{d-k}} \textup{ for all } i \geq 1
\end{align*}
and $\eta \in \nat_+$ is the smallest index such that the total volume of the items in
$\mathcal{M}$ is at most $\epsilon\vol[d]{\mathcal{I}}$.
Clearly $\eta$ is at most $\ceil{1/\epsilon}$.

\subsection{Packing the Large Items}\label{sec:strip packing:large}
We start by creating a packing for the large items into the strip.
The rough procedure for this is the following:
First, we reduce the number of item sizes through rounding.
This creates a constant bound on the number of subsets of rounded large items that can be packed into $A$, as the items with the same size can be considered to be identical.
For each of those subsets, we create a $(d-k)$-dimensional packing into the region of $A$.
Each of those packings will be assigned a $k$-dimensional volume
which will be used to extend this packing to a $d$-dimensional one.
This extension will be accomplished by transforming a $(d-k)$-dimensional item placed in $A$ into an \Nbox{} with enough space for many items of the same size.
As the last step, we stack those packings on top of each other into the strip.

Now we describe the rounding procedure in detail.
We round the sizes of the large items up using the linear grouping
technique \cite{binpacking-aptas} such that the number of
distinct side lengths is a constant $\Nname{sizes}
:= \ceil{1/\paren{\epsilon\rho_\eta^d}}$.
To do this, we partition the large items $\mathcal L$ into groups $\mathcal{G}_1, \dots, \mathcal{G}_{\Nname{sizes}}$ as follows.
We first sort the items in decreasing order of their side lengths.
Let $l=\Nname{sizes}-\abs{\mathcal L}\%\Nname{sizes}$ where $\%$ is the modulo operator.
The first group $\mathcal{G}_1$ contains the first $m := \floor{\abs{\mathcal L}/\Nname{sizes}}$ items,
$G_2$ contains the next $m$ items, and so on until group $G_l$.
Groups $G_{l+1}$ through $G_{\Nname{sizes}}$ contain $m+1$ items each.

The side length of every item in group $\mathcal{G}_i$ for each $i\in[\Nname{sizes}]$
is increased to $\hat s(\mathcal{G}_i)$.
Let the set of rounded up items be denoted by $\mathcal{L}_\up$.
The next lemma shows that this rounding only has a limited effect on the volume and the height of an optimal packing for the large items.

\begin{figure}[h]
  \centering
  \includesvg[width=0.8 \linewidth]{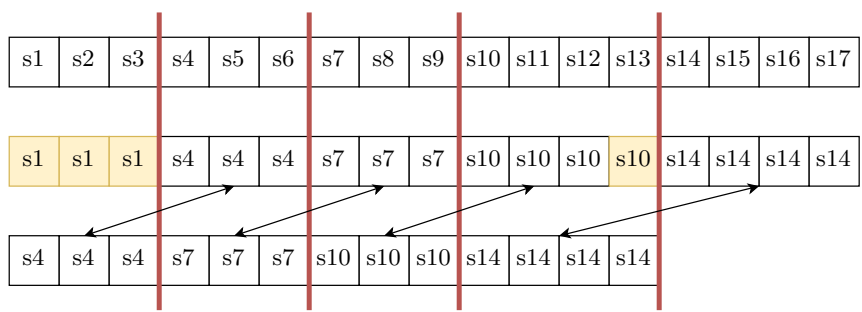}
  \caption{In the first row are the items of $\mathcal L$ sorted in decreasing order of their side lengths.
  The red lines mark the partition into $\mathcal{G}_1$, $\mathcal{G}_2$, $\mathcal{G}_3$, $\mathcal{G}_4$ and $\mathcal{G}_5$ from left to right.
  The second row contains $\mathcal{L}_\up$ and the third row contains $\mathcal{L}_\down$.
  The arrows connect items of $\mathcal{L}_\up$ and $\mathcal{L}_\down$ with the same size.
  Yellow items of $\mathcal{L}_\up$ do not have a corresponding item in $\mathcal{L}_\down$.}
  \label{fig:strip packing:large:rounding}
\end{figure}

\begin{lemma}\label{strip packing:large:rounding}
The volume of $\mathcal{L}_\up$ and the height of an optimal packing for $\mathcal{L}_\up$ are bounded by
\begin{align*}
&\opts{\mathcal{L},A,B}
&\leq& \opts{\mathcal{L}_\up,A,B}
&\leq& \paren{1+2^{k-1}\epsilon}\opts{\mathcal{L},A,B} + \hat{s} \quad\text{and}\\
&\vol{\mathcal L}
&\leq& \vol{\mathcal{L}_\up}
&\leq& \paren{1+\epsilon}\vol{\mathcal L} + \vol[k-1]{B}\hat{s}.\\
\end{align*}
\end{lemma}
\begin{proof}
Round the sizes of the items in $\mathcal L$ down as follows: For $i\in[\Nname{sizes}-1]$,
the size of each item in $\mathcal{G}_i$ is rounded down to the size of the largest item in the next group, i.e., $\hat s(\mathcal{G}_{i+1})$.
The items in $\mathcal{G}_{\Nname{sizes}}$ are discarded.
Let this rounded-down instance of $\mathcal L$ be denoted by $\mathcal{L}_\down$.

The strip packing problems with item sets $\mathcal{L}_\up$ and $\mathcal{L}_\down$ are similar but
$\mathcal{L}_\up$ has at most $m + 1$ more items of size at most $\hat{s}$ than $\mathcal{L}_\down$ as it can be seen in \cref{fig:strip packing:large:rounding}.
Note that
\begin{align}
m &\leq \frac{\abs{\mathcal{L}}}{\Nname{sizes}}\tag{definition of $m$}\nonumber\\
&\leq \epsilon \abs{\mathcal{L}} \rho_\eta^d\tag{definition of $\Nname{sizes}$}\nonumber\\
&\leq \epsilon \vol{\mathcal{L}}\label{m-bound}
\end{align}
Hence, since we have $\hat{s} \leq 1 \leq \vol[k-1]{B}$ from the problem definition,
\begin{align*}
\vol{\mathcal{L}}
\leq \vol{\mathcal{L}_\up}
&\leq \vol{L_\down} + \paren{m+1}\hat{s}^d\\
&\leq \paren{1+\epsilon}\vol{\mathcal{L}} + \vol[k-1]{B}\hat{s}
\end{align*}

Next, we transform an optimal packing for $\mathcal{L}_\down$ into a packing for $\mathcal{L}_\up$ by stacking layers of height $\hat{s}$ on top of it
where we pack the items from $\mathcal{G}_1$ and one item from the first group $\mathcal G_{l+1}$ with $m+1$ items if there is any.
In Kenyon and R{\'{e}}mila's algorithm, the items are stacked on top of each other, so each layer would only contain a single item.
This is justified, because in their case, large items are large compared to the whole base of the strip.
In our case, this might not be true because of the large size of $B$, thus
we are going to use \cref{n-box} to pack
$\Nname{L} := \prod_{i=1}^{k-1}\floor{b_i/\hat{s}}$
items in each layer \kvn{(see \cref{fig:L-up-from-L-down})}.
Let $\abs{\mathcal L}/\Nname{sizes} = q \Nname{L} + r$ with $q \in \nat_0$ and $0 \leq r < N_L$.
\kvn{The number of items that we need to insert is at most
$\ceil{\abs{\mathcal L} / \Nname{sizes}}$.
So} the number of layers that we need is at most
\begin{align}
\ceil{\frac{\ceil{\abs{\mathcal L} / \Nname{sizes}}}{N_L}}
&= \ceil{q + \frac{\ceil{r}}{N_L}}
\leq q + 1
\leq \frac{\floor{\abs{\mathcal L}/\Nname{sizes}}}{N_L} + 1
= \frac{m}{N_L} + 1 \tag{since $\ceil{r} \leq \Nname{L} \in \nat$}\nonumber\\
&= \frac{m}{\prod_{i=1}^{k-1}\floor{b_i/\hat{s}}} + 1\nonumber\\
&\leq \frac{m}{\prod_{i=1}^{k-1}\floor{b_i}} + 1\tag{since $\hat s\le1$}\nonumber\\
&< \frac{m}{\prod_{i=1}^{k-1}\paren{b_i/2}} + 1\tag{since $b_i \geq 1$}\nonumber\\
&= \frac{2^{k-1}m}{\vol[k-1]{B}} + 1\label{NL-bound}
\end{align}
\begin{figure}[h]
\begin{center}
  \includesvg[width=0.3\linewidth]{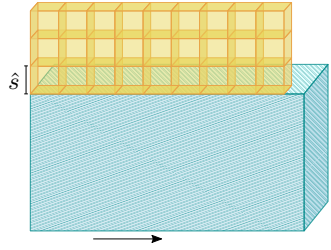}
\end{center}
\caption{
  \kvn{
    Constructing the packing of $\mathcal L_\up$ from a packing of $\mathcal L_\down$ (shown as a shaded cuboid).
    We create a new cuboidal cell of side length $\hat s$ for each item in $\mathcal L_\up\backslash \mathcal L_\down$. We keep creating the
    cells in this way along the length of each side in $B$ adding a new layer on the top if necessary.
    In the figure, there is only one long dimension (shown by the arrow mark).
  }
}
\label{fig:L-up-from-L-down}
\end{figure}

The inequality $\opts{\mathcal{L},A,B}\leq \opts{\mathcal{L}_\up,A,B}$ is easy to see
because any packing of $\mathcal{L}_\up$
can be transformed into a packing of $\mathcal L$ by transforming each hypercube of $\mathcal{L}_\up$ into its corresponding hypercube of $\mathcal L$.
Finally, by equation (\ref{NL-bound}) and the inequality $\vol[d]{\mathcal{L}}\leq\vol[d-k]{A}\vol[k-1]{B}\opts{\mathcal{L},A,B} = \vol[k-1]{B}\opts{\mathcal{L},A,B}$,
\begin{align*}
\opts{\mathcal{L}_\up,A,B}
&\leq \opts{\mathcal{L}_\down,A,B}
+ \paren{\frac{2^{k-1}m}{\vol[k-1]{B}} + 1}\hat{s}\\
&\leq \opts{\mathcal{L},A,B}
+ 2^{k-1} \epsilon \frac{\vol{\mathcal{L}}}{\vol[k-1]{B}}
+ \hat{s}\tag{by equation (\ref{m-bound})}\\
&\leq \paren{1+2^{k-1}\epsilon}\opts{\mathcal{L},A,B} + \hat{s}.&&\qedhere
\end{align*}
\end{proof}

The previous rounding allows us to handle all large items of any group exactly
the same manner as they have the same size and since the rounding doesn't affect
the volume and optimal packing height too much.

Let's now define a \emph{configuration} as a
tuple of natural numbers in $\nat_0^{\Nname{sizes}}$ as follows:
A configuration $\paren{c_1, \dots, c_{\Nname{sizes}}}$ denotes a set of $c_i$ items from group $\mathcal{G}_i$
for all $1\leq i \leq \Nname{sizes}$.
We call a configuration \emph{valid} if the projection of the items represented by it in the first $(d-k)$-dimensions can be packed into the grid $\bigcup_{a\in A}a$;
otherwise, it is invalid.
As each large item has size at least $\rho_\eta$, we know that every configuration with
any tuple value $c_i > 1/ \rho_\eta^{d-k}$ is invalid since the volume of $A$ is one unit.
Thus, the number of valid configurations is upper bounded by
$\paren{1/ \rho_\eta^{d-k} + 1}^{\Nname{sizes}}$ which is a constant.
Further, each of those configurations contains only a constant number of items;
Thus, we can find in a constant time a lower-justified packing for its items into $\bigcup_{a \in A} a$ to verify that it is valid or detect that no such packing exists and the configuration is invalid.
\kvn{
  Lower-justified packing denotes a packing where each item is pushed downward in each dimension
  until it is tangent to a lower item or a lower border of $A$.
  Thus, we can verify if a configuration is valid or not in constant time.
}
Let $\Nname{configs}$ be the number of valid configurations and $c_{i,j}$ be the number of items from group $\mathcal{G}_j$
that configuration number $i$ contains.
We solve the following LP to assign each valid configuration a certain amount of $k$-dimensional volume $x_i$:
\begin{align*}
\min &\sum_{i=1}^{\Nname{configs}} x_i\\
\textrm{subject to}&\sum_{i=1}^{\Nname{configs}} x_i c_{i,j} = \abs{\mathcal{G}_j}\hat{s}\paren{\mathcal{G}_j}^k\quad \textrm{for all }j\in\Nname{sizes}\\
&x_i \geq 0
\end{align*}
This LP has a constant number of variables and constraints, therefore it is solvable in constant time.
This LP can easily seen to be bounded and in the following lemma, we will also prove that it is feasible.
So, there exists an optimal solution and we denote it with $x_i^*$ for $i \in \Set{1, \dots, \Nname{configs}}$.
For each valid configuration number $i$ we create a layer $L_i := \cont{A,B,h_i}$
with height $h_i := x_i^* / \vol[k-1]{B}$.
The following lemma shows that the total height of these layers
does not exceed the optimal height for packing the large items.

\begin{figure}[h]
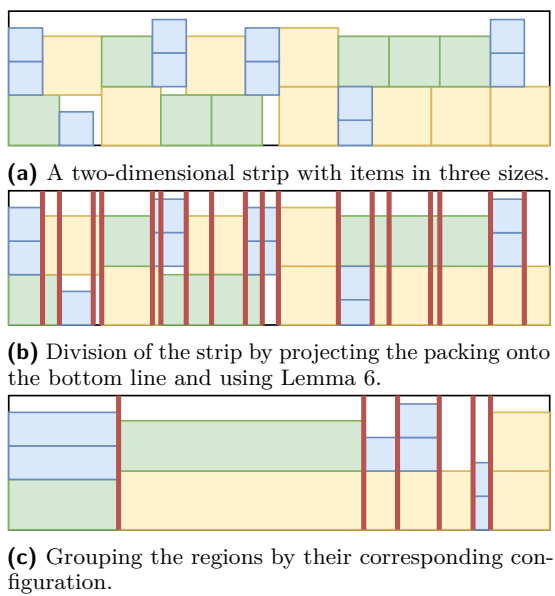

  \centering
  \begin{subfigure}{0.51\textwidth}
    \centering
    \includesvg[width=\linewidth]{img/strip_large_lp}
    \caption{A two-dimensional strip with items in three sizes.}
    \label{fig:strip packing:large packing:lp:1}
  \end{subfigure}
  \hspace{5mm}
  \begin{subfigure}{0.51\textwidth}
    \centering
    \includesvg[width=\linewidth]{img/strip_large_lp_1}
    \caption{Division of the strip by projecting the packing onto the bottom line and using \cref{grid:splitting}.}
    \label{fig:strip packing:large packing:lp:2}
  \end{subfigure}
  \hspace{5mm}
  \begin{subfigure}{0.51\textwidth}
    \centering
    \includesvg[width=\linewidth]{img/strip_large_lp_2}
    \caption{Grouping the regions by their corresponding configuration.}
    \label{fig:strip packing:large packing:lp:3}
  \end{subfigure}
  \caption{Transformation from a packing to a solution of the LP. For this strip $A$ is the vertical dimension, $B$ is zero-dimensional and the height \kvn{of the strip is
  in the horizontal direction}.}
  \label{fig:strip packing:large packing:lp}
\end{figure}
\arir{Should we mention A is 1D and B is 0D in the figure. Added.}

\begin{lemma}\label{strip packing:large:lp}
The total height of the layers $L_1, \ldots, L_{\Nname{configs}}$ can be bounded as
\begin{align*}
\sum_{i=1}^{\Nname{configs}}h_i=\sum_{i=1}^{\Nname{configs}} \frac{x_i^*}{\vol[k-1]{B}} \leq \opts{\mathcal{L}_\up,A,B}
\end{align*}
\end{lemma}
\begin{proof}
Let's consider an optimal strip packing of $\mathcal L_\up$ like in \cref{fig:strip packing:large packing:lp:1}.
We can find some of the configurations of large items that we defined above in this packing.
\mar{If we fix a point in $p \in B \times \intv{0,h}$ and take a look at the intersection $A \times \Set{p}$ with the strip,
we gain a $\paren{d-k}$-dimensional region with the shape of $A$.} \arir{not clear. Better to rephrase. Done.}
In \cref{fig:strip packing:large packing:lp:1} this would be a vertical line.
We are also intersecting some packed items.
As we have a valid packing and we have fixed coordinates $p$ in the dimensions of $B$ and the height,
those items must not intersect in the dimensions of $A$ and therefore form a valid $\paren{d-k}$-dimensional packing in $A$.
Thus, our point $p$ corresponds to a configuration of large items packed in $A$.
We will now project the packing onto $B\times\intv{0,h}$.
Next, we interpret $B\times\intv{0,h}$ as a grid of a single cell and use \cref{grid:splitting} to split it into a set of cells $R$ (without removing the covered cells) as it can be seen in \cref{fig:strip packing:large packing:lp:2}.
All the points in a cell $r\in R$ intersect the same items $\mathcal{L}_r$ of our packing and thus we can assign a configuration number $i_r$ to that cell.
Each item $x \in \mathcal{L}_r$ covers a volume of $s\paren{x}^{d-k}\vol[k]{r}$
\kvnr{I think it should be $s\paren{x}^{d-k}\vol[k]{r}$. Earlier, it was $s\paren{x}\vol[k]{r}$.
So changed.
Please check.}
inside of $A\times r$.
We define $x_i := \sum_{r\in R, i_r=i} \vol[k]{r}$ to be the total volume of the regions assigned to configuration number $i$ which is the width in \cref{fig:strip packing:large packing:lp:3}.

As we have a valid packing of $\mathcal{L}_\up$ into the strip, we have for all $1 \leq j \leq \Nname{sizes}$ the whole volume of the items inside of $\mathcal{G}_j$ packed.
If we use $c_{i,j}$ from the definition of the linear program,
we can see that those $x_i$ are a valid solution to the linear program.

\begin{align*}
\quad \sum_{i=1}^{\Nname{configs}} x_i c_{i,j}
&= \sum_{r\in R} c_{i_r,j} \vol[k]{r}\\
&= \frac{1}{\hat{s}\paren{\mathcal{G}_j}^{d-k}}
\sum_{r\in R} c_{i_r,j} \vol[k]{r} \hat{s}\paren{\mathcal{G}_j}^{d-k}\\
&= \frac{1}{\hat{s}\paren{\mathcal{G}_j}^{d-k}}\abs{\mathcal{G}_j}\hat{s}\paren{\mathcal{G}_j}^d
\tag{packed volume is equal to total volume}\\
&= \abs{\mathcal{G}_j}\hat{s}\paren{\mathcal{G}_j}^k.
\end{align*}
Now we can use the fact that the volume of the regions $\sum_{r\in R} \vol[k]{r}$ is equal to the volume of the strip $\vol[k-1]{B}\opts{\mathcal{L}_\up,A,B}$ to end this proof:
\begin{align*}
\vol[k-1]{B}\opts{\mathcal{L}_\up,A,B}
= \sum_{r\in R}\vol[k]{r}
= \sum_{i=1}^{\Nname{configs}} x_i
\geq \sum_{i=1}^{\Nname{configs}} x_i^*. &\qedhere
\end{align*}
\end{proof}

For each large item $x$ of configuration number $i$ we take its position $\paren{p_1\paren{x}, \dots, p_{d-k}\paren{x}}$ of the lower-justified packing of the configuration into $\bigcup_{a \in A} a$ and
create an \Nbox{} for items of its group by extending the region for this item in $k$ dimensions to
\begin{align*}
\paren{\prod_{j=1}^{d-k}\intv{p_j, p_j + s\paren{x}}}
\times \paren{\prod_{j=1}^{k-1}\intv{0, s\paren{x}\ceil{\frac{b_j}{s\paren{x}}}}}
\times \intv{0, s\paren{x}\ceil{\frac{h_i}{s\paren{x}}}}.
\end{align*}
Those \Nboxes{} are non-overlapping and lie inside of the layer that is assigned to their configuration, after enlarging it by some small value to $L_i' := \cont{A,B+\hat{s},h_i+\hat{s}}$ (see \cref{fig:strip packing:large packing}) where $B + \hat{s}$ is the hypercuboid $B$ after increasing each side length by $\hat{s}$.
By stacking those layers on top of each other we can ensure to pack all large items into $\cont{A,B+\hat{s},\sum_{i=1}^{\Nname{configs}}\paren{h_i+\hat{s}}}$.

\begin{figure}[h]
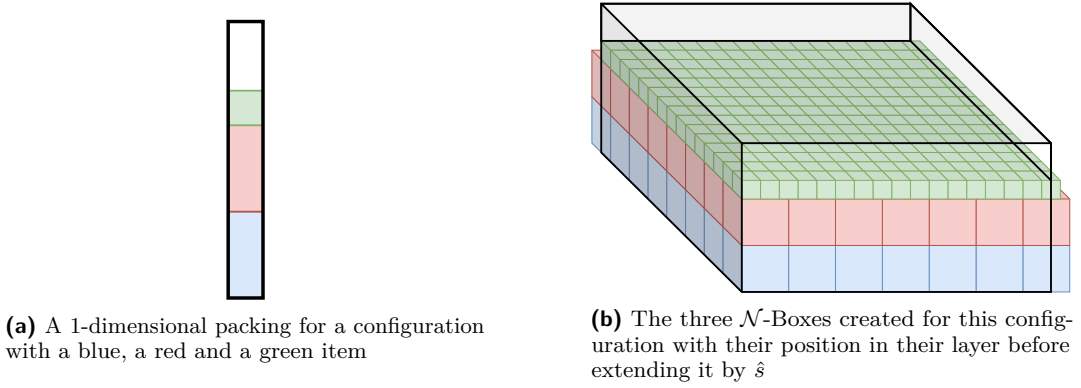

  \centering
  \begin{subfigure}{0.45\textwidth}
    \centering
    \includesvg[width=0.08\linewidth]{img/strip_large_packing_a}
    \caption{A 1-dimensional packing for a configuration with a blue, a red and a green item}
    \label{fig:strip packing:large packing:1}
  \end{subfigure}
  \hfill
  \begin{subfigure}{0.45\textwidth}
    \centering
    \includesvg[width=\linewidth]{img/strip_large_packing_b}
    \caption{The three \Nboxes{} created for this configuration with their position in their layer before extending it by $\hat{s}$}
    \label{fig:strip packing:large packing:2}
  \end{subfigure}
  \caption{3-dimensional packing derived from a 1-dimensional configuration}
  \label{fig:strip packing:large packing}
\end{figure}

\subsection{Packing Small Items in the Gaps}\label{sec:strip packing:gaps}
We now focus on the gaps inside of each layer $L_i$ that are not reserved for the \Nboxes{} that contain the large items.
We want to pack small items in those gaps using NFDH.
Additionally, we want those gaps to match the requirements for \Vboxes{}, therefore we will also perform some rounding on their size.

As there are at most ${1}/{\rho_\eta^{d-k}}$ large items packed into $\bigcup_{a \in A} a$ by the configuration from which this layer was created,
we can split the free space in each individual $a\in A$ into a grid with at most $2/\rho_\eta^{d-k}$ layers in each dimension using \cref{grid:splitting} and the fact that we have a lower-justified packing, which reduces the number of additional grid boundaries by one in each dimension.
Thus, this grid contains at most $2^{d-k}/\rho_\eta^{\left(d-k\right)^2}$ hypercuboidal cells which we
will call gaps. This way, we split the entire $(d-k)$-dimensional region $A$ into a set $G_i$ of
at most $N 2^{d-k}/\rho_\eta^{\left(d-k\right)^2}$ hypercuboidal gaps.
This method also ensures that each gap is contained in some $a\in A$ and therefore of a limited side length of $\alpha \hat{s}$.
We can define a function $\ext$ that extends each gap $g \in G_i$ to a $d$-dimensional gap
$g \times B\times \left[0, h_i \right]$.
Note, that we assume here for simplicity that each layer starts at a height of $0$ and reaches a height of $h_i$, in reality the layers are stacked on top of each other, so this last interval is actually shifted by the heights of the previous enlarged layers $L_1', \dots, L_{i-1}'$.
By the way in which the layer was constructed the gaps $\Set{\ext\paren{g} | g \in G_i}$ partition the empty space of the layer $L_i$.
Consider \cref{fig:strip packing:large packing}.
The packing into $A$ in \cref{fig:strip packing:large packing:1} has a single $1$-dimensional gap on top of the items.
The cells of $A$ are not visualized in this figure, so it might be the case that this gap is divided into multiple gaps because of the boundaries of the cells, but for simplicity let us suppose this gap is not divided.
In \cref{fig:strip packing:large packing:2} we see, how this gap is extended to a $3$-dimensional gap in the same way as the packed items were extended to \Nboxes{} and it covers the whole empty space of the layer.

We can also define another function $\ext'$ that rounds the side lengths of each $g \in G_i$ down to the largest
multiple of $\hat{s}\paren{\mathcal{S}}$ before applying $\ext$.
The following lemma shows that we don't lose too much of the packable volume in the gap by rounding.

\begin{lemma}\label{strip packing:gaps:rounding}
Let $g \in G_i$ be a gap in layer $L_i$.
The difference between the volumes of its rounded and non-rounded extensions
can be bounded by
\begin{align*}
\vol{\ext\paren{g}}-\vol{\ext'\paren{g}}
&\leq
\rho_{\eta+1} \paren{d-k} x_i^* \paren{\alpha\hat{s}}^{d-k-1}.
\end{align*}
\end{lemma}
\begin{proof}
The hypercuboidal region $\ext'\paren{g}$ is obtained from $\ext\paren{g}$ by decreasing the
side lengths in $(d-k)$ dimensions by a value of at most $\hat{s}\paren{\mathcal S}$. Let
$g_1,g_2,\dots,g_{d-k}$ be the lengths of $\ext\paren{g}$ in those dimensions. Then
\begin{align*}
\vol{\ext'\paren{g}}&\ge \vol[k-1]{B}h_i
				\prod_{j=1}^{d-k}(g_j-\hat{s}\paren{\mathcal S})\\
		&= x_i^*\prod_{j=1}^{d-k}(g_j-\hat{s}\paren{\mathcal S})\tag{definition of $h_i$}\\
		&\ge x_i^*\left(\prod_{j=1}^{d-k}g_j-\hat{s}\paren{\mathcal S}\sum_{j=1}^{d-k}
				\prod_{\substack{k\in[d-k]\\k\ne j}}g_k\right)
		\tag{inductive argument over $(d-k)$}\\
		&\ge x_i^*\left(\prod_{j=1}^{d-k}g_j-(d-k)\hat{s}\paren{\mathcal S}(\alpha\hat s)^{d-k-1}\right)
					\tag{since for all $i\in[d-k]$, $g_i\le \alpha\hat s$}\\
		&\ge x_i^*\prod_{j=1}^{d-k}g_j- \rho_{\eta+1} \left(d-k\right) x_i^* \left(\alpha\hat{s}\right)^{d-k-1}
			\tag{since $\hat{s}\paren{\mathcal S}\le \rho_{\eta+1}$}\\
		&= \vol{\ext\paren{g}}-\rho_{\eta+1} \left(d-k\right) x_i^* \left(\alpha\hat{s}\right)^{d-k-1}&&&&&&\;\;\qedhere
\end{align*}
\end{proof}

For each gap $g \in G_i$ we now assign a \kvn{maximal} set of small items $\mathcal{I}_{i,g}$ such that the volume of these items does not exceed
$$\vol{\ext'\paren{g}} - \rho_{\eta+1}\frac{\surf{\ext'\paren{g}}}{2}$$
By doing so, we ensure that the items in $\mathcal{I}_{i,g}$ can be packed in \kvn{$\ext'(g)$} using \cref{v-box}.
We can now distinguish between two cases, where we either (i) pack all small items in this way
or (ii) we ensure that small items are packed into each gap \kvn{$\ext'(g)$} with a volume of at least $$\vol{\mathcal{I}_{i,g}} > \vol{\ext'\paren{g}} - \rho_{\eta+1}\frac{\surf{\ext'\paren{g}}}{2} - \rho_{\eta+1}^d$$
as $\rho_{\eta+1}^d$ is an upper bound on the volume of a small item.
\kvn{In the second case, we can show that each layer is sufficiently packed, i.e., the wasted space is small. We show this
in the next lemma.}

\begin{lemma}\label{strip packing:gaps:fvol}
For a layer $L_i$, let $\fvol{L_i}$ denote the volume of the empty space in $L_i$ left after packing small items in the set of gaps \kvn{$\{\ext'(g)|g\in G_i\}$}.
If there are still small items left unpacked, then we have the following inequality.
\begin{align*}
\fvol{L_i} \leq \frac{3}{2}\epsilon x_i^* + \frac{\epsilon}{d-1}\vol[k-1]{B}\hat{s}.
\end{align*}
\end{lemma}
\begin{proof}
By using the upper bound $\alpha\hat{s}$ for the side length of any gap $g \in G_i$ in the first $(d-k)$ dimensions
and since $b_j \geq 1$ we can bound the surface of the rounded gaps.
\kvn{Let $g_1,g_2,\dots,g_{d-k}$ denote the side lengths of $\ext'(g)$.}
\kvnr{Added two equalities for clarity. Check once.}
\begin{align*}
&\quad\frac{\surf{\ext'\paren{g}}}{2}\\
&\leq \frac{\surf{\ext\paren{g}}}{2}\\
&=\frac{\vol{\ext'(g)}}{2}\left(\sum_{j=1}^{k-1}\frac{1}{b_j}+\sum_{j=1}^{d-k}\frac{1}{g_j}+\frac{\vol[k-1]{B}}{x_i^*}\right)\tag{by definition of $\surf{\cdot}$}\\
&=\frac{g_1g_2\dots g_{d-k}x_i^*}{2}\left(\sum_{j=1}^{k-1}\frac{1}{b_j}+\sum_{j=1}^{d-k}\frac{1}{g_j}+\frac{\vol[k-1]{B}}{x_i^*}\right)\\
&\leq \paren{d-k}x_i^*\paren{\alpha\hat{s}}^{d-k-1}
+ \paren{\alpha\hat{s}}^{d-k}x_i^*\sum_{j=1}^{k-1}\frac{1}{b_j}
+ \paren{\alpha\hat{s}}^{d-k}\vol[k-1]{B}\tag{as each $g_i\le \alpha \hat s$}\\
&\leq \paren{d-k}x_i^*\paren{\alpha\hat{s}}^{d-k-1}
+ \paren{\alpha\hat{s}}^{d-k}\paren{x_i^*\paren{k-1}
+ \vol[k-1]{B}}\tag{since each $b_i\ge 1$}
\end{align*}
Now we can bound the free volume inside of a gap $g \in G_i$:
\begin{align*}
&\quad\vol{\ext\paren{g}} - \vol{\mathcal{I}_{i,g}}\\
&\leq \vol{\ext'\paren{g}} - \vol{\mathcal{I}_{i,g}}
+ \rho_{\eta+1} \paren{d-k} x_i^* \paren{\alpha\hat{s}}^{d-k-1}
\tag{by \cref{strip packing:gaps:rounding}}\\
&< \rho_{\eta+1} \frac{\surf{\ext'\paren{g}}}{2}
+ \rho_{\eta+1}^d
+ \rho_{\eta+1} \paren{d-k} x_i^* \paren{\alpha\hat{s}}^{d-k-1}
\tag{\kvn{using the lower bound on $\vol{\mathcal{I}_{i,g}}$}}\\
&\leq
2\rho_{\eta+1} \paren{d-k} x_i^* \paren{\alpha\hat{s}}^{d-k-1}
+ \rho_{\eta+1}\paren{\alpha\hat{s}}^{d-k}
\paren{x_i^*\paren{k-1} + \vol[k-1]{B}}
+ \rho_{\eta+1}^d.
\end{align*}
Finally, we add the free volume of all the gaps in layer $L_i$
and use the bound $\abs{G_i} \leq N 2^{d-k}/\rho_\eta^{\left(d-k\right)^2}$:
\begin{align*}
&\quad\fvol{L_i}\\
&= \sum_{g\in G_i}\paren{\vol{\ext\paren{g}} - \vol{\mathcal{I}_{i,g}}}\\
&\leq
\frac{N {2^{d-k}}}{\rho_\eta^{\paren{d-k}^2}}
\rho_{\eta+1}
\paren{
2 \paren{d-k} x_i^* \paren{\alpha\hat{s}}^{d-k-1}
+\paren{\alpha\hat{s}}^{d-k}
\paren{x_i^*\paren{k-1} + \vol[k-1]{B}}
+ \rho_{\eta+1}^{d-1}
}\\
&\leq
\frac{\epsilon\hat{s}}{2\paren{d-1}\paren{\alpha\hat{s}}^{d-k}}
\paren{
2 \paren{d-k} x_i^* \paren{\alpha\hat{s}}^{d-k-1}
+\paren{\alpha\hat{s}}^{d-k}
\paren{x_i^*\paren{k-1} + \vol[k-1]{B}}
+ \rho_{\eta+1}^{d-1}
}
\tag{by definition of $\rho_{\eta+1}$ and $\rho_1$}\\
&\leq
\frac{\epsilon}{\alpha}x_i^*
+ \frac{\epsilon\hat{s}}{2}x_i^*
+ \frac{\epsilon\hat{s}}{2\paren{d-1}}\vol[k-1]{B}
+ \frac{\epsilon\hat{s}}{2\paren{d-1}\paren{\alpha\hat{s}}^{d-k}}\rho_{\eta+1}^{d-1}\\
&\leq
\frac{3}{2}\epsilon x_i^*
+ \frac{\epsilon}{2\paren{d-1}}\vol[k-1]{B}\hat{s}
+ \frac{\epsilon}{2\paren{d-1}\alpha^{d-k}}\rho_{\eta+1}^{k-1}\hat{s}
        \tag{since $\alpha\ge 1$ and $\rho_{\eta+1}\le \hat{s}\leq 1$}\\
&\leq
\frac{3}{2}\epsilon x_i^*
+ \frac{\epsilon}{d-1}\vol[k-1]{B}\hat{s}
              \tag{since $\rho_{\eta+1}\le1$, $\alpha \ge 1$ and $\vol[k-1]{B}\ge1$}
\end{align*}
\end{proof}
Because we enlarged the layers at the end of \cref{sec:strip packing:large} in the long dimensions,
we can enlarge each gap
\kvn{in the long dimensions and along the height of the strip}
such that each side length is a multiple of $\hat{s}\paren{S}$ and the gap is thus a valid \Vbox{}.

\subsection{Packing Medium and the Remaining Small Items on the Top}\label{sec:strip packing:top}

The remaining unpacked small items are packed together with the medium items in a new layer $L_{\text{top}}$ on top of the strip using NFDH.
We use $\hat{s}\paren{\mathcal{M}} := \hat{s}\paren{\mathcal{S}}$ in the special case where there are no medium items.
Contrary to \cite{harren-journal}, where a single hypercuboid forms the base of the strip
and the items of the additional layer are just packed using NFDH,
our base is formed by multiple hypercuboids.
Each of those hypercuboids forms its own strip and we have to decide for each item in which strip to pack it to maintain a low height.
Moreover, it might be the case that the largest medium item does not fit in any of those individual strips and we have to combine strips to pack this item.
Let us first discuss how to pack items into a strip with a hypercuboid as the base.

\begin{lemma}\label{strip packing:top layer}
Let $c$ be a $(d-1)$-dimensional hypercuboid and $\mathcal{I}$ be a set of $d$-dimensional hypercubes of size at most $\rho_{\eta}$.
Let $2 \vol[d-1]{c} > \rho_{\eta}\surf[d-1]{c}$.
The items can then be packed using NFDH into a strip with base $c$ and height
\begin{align*}
h = \frac{\vol{\mathcal{I}} + \rho_{\eta} \vol[d-1]{c}}
{\vol[d-1]{c} - \rho_{\eta}\surf[d-1]{c}/2}.
\end{align*}
\end{lemma}
\begin{proof}
We know by \cref{nfdh}, that we can either pack all items into the hypercuboid defined by $g := c \times \intv{0, h}$ or we can pack a subset $\mathcal{I}' \subseteq \mathcal{I}$ of the items leaving free space inside of this hypercuboid of volume
\begin{align*}
\vol[d-1]{c}h - \vol{\mathcal{I}'}
&= \vol{g} - \vol{\mathcal{I}'}\\
&\leq \rho_{\eta}\frac{\surf{g}}{2}\\
&= \rho_{\eta}\paren{\frac{\surf[d-1]{c}}{2}h + \vol[d-1]{c}}.
\end{align*}
We may restate this inequality as
\begin{align*}
\paren{\vol[d-1]{c} - \rho_{\eta}\frac{\surf[d-1]{c}}{2}}h
\leq \vol{\mathcal{I}'} + \rho_{\eta}\vol[d-1]{c},
\end{align*}
and that is equivalent to
\begin{align*}
\vol{\mathcal{I}} + \rho_{\eta} \vol[d-1]{c}
\leq \vol{\mathcal{I}'} + \rho_{\eta}\vol[d-1]{c}.
\end{align*}
This finally yields $\vol{\mathcal{I}} \leq \vol{\mathcal{I}'}$ by using the definition of $h$
and since $\mathcal{I}'$ is a subset of $\mathcal{I}$, we also have $\vol{\mathcal{I}}\geq \vol{\mathcal{I}'}$ and thus all items are packed.
\end{proof}

Next, we want to ensure that we can use the packing method described in \cref{strip packing:top layer},
meaning for each item there is a strip that is large enough to pack the item into.
\kvn{
  To do this, we temporarily pack a $(d-k)$ dimensional cube of side length $\hat s$ on $A$
  (see \cref{fig:strip packing:top layer:1})}
and take the set of all intersected hypercuboids $A_{\text{merged}} \subseteq A$ (see \cref{fig:strip packing:top layer:2}).
Those hypercuboids form a larger hypercuboid with side lengths at least $\hat{s}$ in all $(d-k)$ dimensions.

This is because $A$ is assumed to be a grid.
If we assume, that the grid boundaries are named $g_{i,j}$ for $i \in [d]$ and some values of $j$, then we find for each dimension $i$ a maximal index $j_{i,-}$ and a minimal index $j_{i,+} > j_{i,-}$ such that $g_{i,j_{i,-}}$ is smaller than the lowest point of the packed large item in \kvn{the $i\Th$} dimension and $g_{i,j_{i,+}}$ is larger than the highest point of the packed large item in \kvn{the $i\Th$} dimension because the item is packed inside of the grid.
Furthermore the cells $\Set{\prod_{i=1}^d\intv{g_{i,j_i}, g_{i,j_i+1}} | j_{i,-} \leq j_i < j_{i,+} \text{ for all }i\in[d]}$ are exactly the cells of $A$ that are intersected by the packed item.
Those cells have to be part of the grid $A$ because otherwise the large item would not be packed into the grid.
Combining those cells creates the hypercuboid $\prod_{i=1}^d\intv{g_{i,j_{i,-}}, g_{i,j_{i,+}}}$

Next, we round the side lengths of each other hypercuboid in $A\backslash A_{\text{merged}}$ down to the next multiple of $\hat{s}\paren{\mathcal M}$
and we round the side lengths of the new hypercuboid $\bigcup_{a \in A_{\text{merged}}} a$ down to the next multiple of $\hat{s}\left(\mathcal M\right)$ to create rounded instances for the hypercuboids in $A_{\text{merged}}$.
This creates a rounded set $A'$ of hypercuboids where $A_{\text{merged}}' \subseteq A'$ are the rounded instances from the hypercuboids in $A_{\text{merged}}$ (see \cref{fig:strip packing:top layer:3}).
Note that each hypercuboid in $A'$ is formed from a unique hypercuboid in $A$ by reducing the size in each dimension by at most $\hat{s}\left(\mathcal M\right)$,
this is also true for the hypercuboids in $A_{\text{merged}}'$.
This allows us to use similar dependencies between the rounded and non-rounded instances concerning volume and surface,
as we used for the gaps in \cref{sec:strip packing:gaps}.
As a last step, we replace each $\left(d-k\right)$-dimensional hypercuboid $a$ inside of $A$ or $A'$ by its $\left(d-1\right)$-dimensional
extension $a \times \prod_{j=1}^{k-1} [0, b_j]$
and define $A'' :=  A'\backslash A_{\text{merged}}' \cup \Set{\bigcup_{a \in A_{\text{merged}}'} a}$ by merging the hypercuboids in $A_{\text{merged}}'$ to a larger hypercuboid as previously announced (see \cref{fig:strip packing:top layer:4}).
Now we know that there is at least one base for a strip,
where we can pack any of the medium and small items in.
\begin{figure}[h]
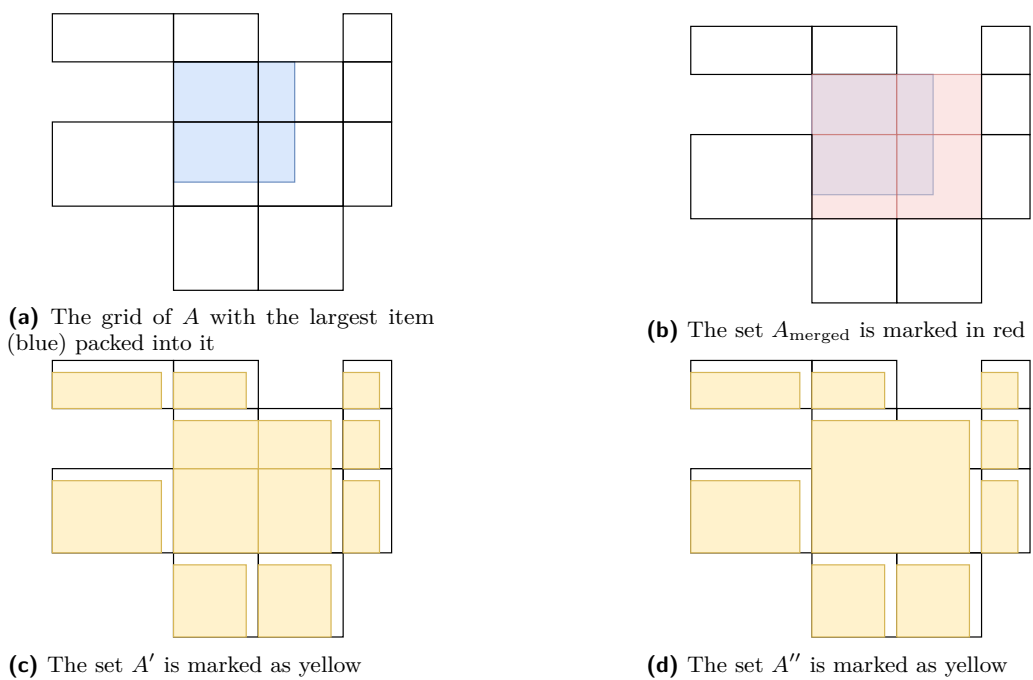

  \centering
  \begin{subfigure}{0.4\textwidth}
    \centering
    \includesvg[width=0.8\linewidth]{img/top_layer_a.svg}
    \caption{The grid of $A$ with the largest item (blue) packed into it}
    \label{fig:strip packing:top layer:1}
  \end{subfigure}
  \hfill
  \begin{subfigure}{0.4\textwidth}
    \centering
    \includesvg[width=0.8\linewidth]{img/top_layer_a_merged.svg}
    \caption{The set $A_{\text{merged}}$ is marked in red}
    \label{fig:strip packing:top layer:2}
  \end{subfigure}
  \begin{subfigure}{0.4\textwidth}
    \centering
    \includesvg[width=0.8\linewidth]{img/top_layer_a_1.svg}
    \caption{The set $A'$ is marked as yellow}
    \label{fig:strip packing:top layer:3}
  \end{subfigure}
  \hfill
  \begin{subfigure}{0.4\textwidth}
    \centering
    \includesvg[width=0.8\linewidth]{img/top_layer_a_2.svg}
    \caption{The set $A''$ is marked as yellow}
    \label{fig:strip packing:top layer:4}
  \end{subfigure}
  \caption{\kvn{Transforming the grid $A$ to make it suitable for packing the small and medium items.}}
  \label{fig:strip packing:top layer}
\end{figure}

\begin{lemma}\label{strip packing:top layer:existence}
There exists a hypercuboid $c \in A''$, such that $2 \vol[d-1]{c} > \rho_{\eta} \surf[d-1]{c}$.
\end{lemma}
\begin{proof}
Let $c = \left(\bigcup_{a \in A_{\text{merged}}'} a\right)$ be the hypercuboid we got through merging.
Let $c_1, \dots, c_{d-1}$ be the side lengths of this hypercuboid.
Because of the creation and rounding process, each of those lengths is at least $\hat{s} - \rho_\eta$.
We have $\vol[d-1]{c} = \prod_{j=1}^{d-1}c_j$. So
\begin{align*}
\surf[d-1]{c}
= 2 \vol[d-1]{c} \sum_{j=1}^{d-1} \frac{1}{c_j}.
\end{align*}
Thus, we only have to show that $1 > \rho_{\eta} \sum_{j=1}^{d-1} 1/c_j$. It holds that
\begin{align*}
\sum_{j=1}^{d-1} \frac{\rho_{\eta}}{c_j}
< \sum_{j=1}^{d-1} \frac{\rho_{\eta}}{\hat{s} - \rho_{\eta}}
= \sum_{j=1}^{d-1} \frac{1}{\hat{s}/\rho_{\eta} - 1}
\leq \sum_{j=1}^{d-1} \frac{1}{\hat{s}/\rho_1 - 1}
< \sum_{j=1}^{d-1} \frac{1}{d - 1}
= \frac{d - 1}{d - 1}
= 1.&\qedhere
\end{align*}
\end{proof}

Now we mark each hypercuboid $c \in A''$ for which $2 \vol[d-1]{c} > \rho_{\eta} \surf[d-1]{c}$ holds, knowing that there is at least one, and set it as the base of a strip.
For each base $a \in A''$ we assign a subset $\mathcal{I}_{\text{top},a}$ of the remaining small and medium items and set its height $h_{\text{top},a}$ according to \cref{strip packing:top layer}.
\kvn{
  That is, we initialize $h_{\text{top},a}$ as
  \begin{align*}
  h_{\text{top},a}=\frac{\rho_{\eta} \vol[d-1]{c}}
{\vol[d-1]{c} - \rho_{\eta}\surf[d-1]{c}/2}
  \end{align*}
}
Initially all sets $\mathcal{I}_{\text{top},a}$ are empty.
Note that this still leads to a nonzero initial height for each strip.
For each item $x$ that needs to be assigned, we assign it to a marked base $a\in A''$ for which the new height $h_{\mathrm{top},a}' = h_{\mathrm{top},a} + \vol{x}/
\paren{\vol[d-1]{c} - \rho_{\eta}\surf[d-1]{c}/2}$ would be minimal to ensure that those heights stay balanced.

After all items are assigned, we define $h_{\mathrm{top}}$ to be the maximal height of any strip where we packed an item in
and define $\ext_{\mathrm{top}}$ as the extension of a $\paren{d-1}$-dimensional hypercuboid $a$ to the $d$-dimensional hypercuboid $a \times \intv{0,h_{\mathrm{top}}}$ of height $h_{\mathrm{top}}$.
We pack $\mathcal{I}_{top,a}$ into $\ext_{\mathrm{top}}\paren{a}$ using NFDH for all $a \in A''$ and create this way a new layer $L_{\mathrm{top}}$ of height $h_{\mathrm{top}}$.
Now we can analyze the free volume of each individual box $\ext_{\mathrm{top}}\paren{a}$.

\begin{lemma}\label{strip packing:top layer:fvol:single}
For all $a \in A''$ the free volume inside of the box $\ext_{\mathrm{top}}\paren{a}$ is at most
\begin{align*}
\fvol{\ext_{\mathrm{top}}\paren{a}}
\leq \rho_{\eta}\frac{\surf{\ext_{\mathrm{top}}\paren{a}}}{2} + \rho_{\eta}^d
\end{align*}
\end{lemma}
\begin{proof}
For the unmarked hypercuboids $a \in A''$ the statement is trivial,
as they have $\vol[d-1]{a} \leq \rho_{\eta} \frac{\surf[d-1]{a}}{2}$ by definition.

Therefore, let $a \in A''$ be marked.
In this case the strip was packed using the method explained above.
Thus, we know that it is packed to a large height.
More precisely, we know through our balancing method,
that adding an item with side length $\rho_\eta$
(which is the upper bound for the side length of a medium item)
would increase the height assigned to that strip through \cref{strip packing:top layer} to a larger value than $h_{\mathrm{top}}$.
\begin{align*}
h_{\mathrm{top}}
&\leq h_{\mathrm{top},a} + \frac{\rho_{\eta}^d}
{\vol[d-1]{a} - \rho_{\eta}\surf[d-1]{a}/2}\\
&= \frac{\vol{\mathcal{I}_{\mathrm{top},a}} + \rho_{\eta}\vol[d-1]{a} + \rho_{\eta}^d}
{\vol[d-1]{a} - \rho_{\eta}\surf[d-1]{a}/2}
\tag{by definition of $h_{\mathrm{top},a}$}
\end{align*}
\kvn{Hence, we obtain that}
\kvnr{Split the equation set into two parts to avoid hbox warnings}
\begin{align*}
\vol[d-1]{a}h_{\mathrm{top}} - \rho_{\eta}\frac{\surf[d-1]{a}}{2}h_{\mathrm{top}}
\leq \vol{\mathcal{I}_{\mathrm{top},a}} + \rho_{\eta}\vol[d-1]{a} + \rho_{\eta}^d\\
\Longleftrightarrow  \fvol{\ext_{\mathrm{top}}\paren{a}}
\leq \rho_{\eta}\paren{\frac{\surf[d-1]{a}}{2}h_{\mathrm{top}} + \vol[d-1]{a}} + \rho_{\eta}^d
\end{align*}
By unfolding the definition for the surface we get
\begin{align*}
\frac{\surf{\ext_{\mathrm{top}}\paren{a}}}{2}
= \frac{\surf[d-1]{a}}{2}h_{\mathrm{top}} + \vol[d-1]{a}.
\end{align*}
This concludes the proof.
\end{proof}
Now we can consider the free volume inside of the whole layer.
\begin{lemma}\label{strip packing:top layer:fvol}
The free volume in layer $L_{\mathrm{top}}$ is at most
\begin{align*}
\fvol{L_{\mathrm{top}}} \leq \frac{3}{2} \epsilon h_{\mathrm{top}}\vol[k-1]{B}
+ \frac{\epsilon}{d-1}\vol[k-1]{B}\hat{s}.
\end{align*}
\end{lemma}
\begin{proof}
We can make a proof similar to the one for \cref{strip packing:gaps:fvol}.
The layer $L_{\mathrm{top}}$ can be split into the hypercuboids $\ext_{\mathrm{top}}\paren{a}$ for $a \in A$.
We start by relating the hypercuboids in $A$ from our splitting to the volume of the hypercuboids in $A''$ which we used for packing.
For the cardinalities we know $\abs{A''} \leq \abs{A'} = \abs{A} \leq N$.
\kvn{
  Also, note that since $A''$ is obtained from $A'$ by merging a few hypercuboids,
  $\sum_{a\in A''}\vol{\ext_{\mathrm{top}}\paren{a}}=\sum_{a\in A'}\vol{\ext_{\mathrm{top}}\paren{a}}$.
  However, $\sum_{a\in A''}\surf{\ext_{\mathrm{top}}\paren{a}}\le\sum_{a\in A'}\surf{\ext_{\mathrm{top}}\paren{a}}$.
}
\begin{align*}
&\sum_{a\in A}\vol{\ext_{\mathrm{top}}\paren{a}}\\
\leq &\sum_{a\in A'}\paren{
\vol{\ext_{\mathrm{top}}\paren{a}}
+ \rho_{\eta} \paren{d-k} h_{\mathrm{top}} \paren{\alpha\hat{s}}^{d-k-1}\vol[k-1]{B}
}
\tag{similar to \cref{strip packing:gaps:rounding} but with $\rho_\eta$ and $h_{\mathrm{top}}\vol[k-1]{B}$ instead of $\rho_{\eta+1}$ and $x_i^*$}\\
\le& N\rho_{\eta}\vol[k-1]{B}\paren{d-k} h_{\mathrm{top}} \paren{\alpha\hat{s}}^{d-k-1}
+ \sum_{a\in A''}\vol{\ext_{\mathrm{top}}\paren{a}}\\
\end{align*}
\kvn{
Similar to how we upper-bounded $\surf{\ext'(g)}/2$ in \cref{strip packing:gaps:fvol}, instead of using $x_i^*$,
if we use $h_{\mathrm{top}}\vol[k-1]{B}$, we obtain
}
\begin{align*}
\frac{\surf{\ext_{\mathrm{top}}\paren{a}}}{2}
\leq \vol[k-1]{B}\paren{
\paren{d-k} h_{\mathrm{top}} \paren{\alpha\hat{s}}^{d-k-1}
+ \paren{\alpha\hat{s}}^{d-k}\paren{h_{\mathrm{top}}\paren{k-1} + 1}
}
\end{align*}
By applying those two inequalities and \cref{strip packing:top layer:fvol:single} we can bound the free volume to:
\begin{align*}
&\quad\fvol{L_{\mathrm{top}}}\\
&= \sum_{a\in A}\vol{\ext_{\mathrm{top}}\paren{a}}
- \sum_{a\in A''}\vol{\mathcal{I}_{top,a}}\\
&\leq N\rho_{\eta}\vol[k-1]{B}\paren{d-k} h_{\mathrm{top}} \paren{\alpha\hat{s}}^{d-k-1}
+ \sum_{a\in A''}\paren{\vol{\ext_{\mathrm{top}}\paren{a}} - \vol{\mathcal{I}_{top,a}}}\\
&\leq N\rho_{\eta}\vol[k-1]{B}\paren{d-k} h_{\mathrm{top}} \paren{\alpha\hat{s}}^{d-k-1}
+ \sum_{a\in A''}\paren{\rho_{\eta}\frac{\surf{\ext_{\mathrm{top}}\paren{a}}}{2}+\rho_{\eta}^d}\tag{\cref{strip packing:top layer:fvol:single}}\\
&\leq N\rho_{\eta}\vol[k-1]{B}\paren{d-k} h_{\mathrm{top}} \paren{\alpha\hat{s}}^{d-k-1}
+ \sum_{a\in A'}\paren{\rho_{\eta}\frac{\surf{\ext_{\mathrm{top}}\paren{a}}}{2}+\rho_{\eta}^d}\\
&\leq N \rho_{\eta} \vol[k-1]{B} \paren{
2\paren{d-k} h_{\mathrm{top}} \paren{\alpha\hat{s}}^{d-k-1}
+ \paren{\alpha\hat{s}}^{d-k} \paren{h_{\mathrm{top}}\paren{k-1} + 1}
}\\
&\qquad\qquad\qquad\qquad\qquad\qquad\qquad\qquad\qquad\qquad\qquad\qquad\qquad\qquad\qquad+N\rho_{\eta}^{d}
\end{align*}
Now we can \kvn{use the fact that $\rho_\eta\le\rho_1$} which is $\epsilon\hat{s}/\paren{2\paren{d-1}N\paren{\alpha\hat{s}}^{d-k}}$.
We use this to simplify the individual summands of this inequality.
\begin{align*}
&N \rho_{\eta} \vol[k-1]{B} 2\paren{d-k} h_{\mathrm{top}} \paren{\alpha\hat{s}}^{d-k-1}
\leq \epsilon h_{\mathrm{top}} \vol[k-1]{B}\\
&N \rho_{\eta} \vol[k-1]{B} \paren{\alpha\hat{s}}^{d-k} h_{\mathrm{top}}\paren{k-1}
\leq \frac{\epsilon}{2} \frac{k-1}{d-1} \hat{s} h_{\mathrm{top}} \vol[k-1]{B}
\leq \frac{\epsilon}{2} h_{\mathrm{top}} \vol[k-1]{B}\\
&N \rho_{\eta} \vol[k-1]{B} \paren{\alpha\hat{s}}^{d-k}
\leq\frac{\epsilon}{2\paren{d-1}} \hat{s} \vol[k-1]{B}\\
&N \rho_{\eta}^{d}
\leq\frac{\epsilon}{2\paren{d-1}} \frac{1}{\alpha^{d-k}} \frac{\rho_{\eta}^{d-k}}{\hat{s}^{d-k}}\rho_{\eta}^{k-1}\hat{s}
\leq\frac{\epsilon}{2\paren{d-1}} \hat{s} \vol[k-1]{B}
\end{align*}
This sums up to the inequality
\begin{align*}
\fvol{L_{\mathrm{top}}} \leq \frac{3}{2} \epsilon h_{\mathrm{top}} \vol[k-1]{B}
+ \frac{\epsilon} {d-1}\vol[k-1]{B}\hat{s}.&\qedhere
\end{align*}
\end{proof}

Now, that all items are packed, we extend the size of the top layer in the long dimensions by $\hat{s}$.
This way, we gain enough space to also increase the size of the boxes inside of the top layer in those dimensions, such that their side lengths are multiples of $\hat{s}\paren{\mathcal M}$ and they are thus \Vboxes{}.

\subsection{Analysis}\label{sec:strip packing:analysis}
Now we can analyze the total height $h_{\mathrm{total}}$ of the packing.
We have to distinguish between the case where all of the small items are packed in the gaps and the case where some of them were packed in the top layer.
If the gaps were all filled with a large volume of small items, we can use volume arguments to find a bound on the height $h_{\mathrm{total}}$.

\begin{lemma}\label{strip packing:case 1}
If not all small items were packed into the gaps of the layers $L_1, \dots, L_{\Nname{configs}}$, then we have a height of
\begin{align*}
h_{\mathrm{total}} \leq \paren{1 + 5\epsilon} \opts{\mathcal{I}, A, B}
+ 2\paren{\Nname{configs} + 2}\hat{s}.
\end{align*}
\end{lemma}
\begin{proof}
Let $\vol[d-k]{P_i} := \sum_{j=1}^{\Nname{sizes}}c_{i,j}\hat{s}\paren{\mathcal{G}_j}^{d-k}$ be the volume of the $\paren{d-k}$-dimensional packing $P_i$ associated with configuration $i$.
By \cref{strip packing:large:rounding} we can bound the volume of the large items:
\begin{align*}
\vol{\mathcal L_\up}
&= \sum_{j=1}^{\Nname{sizes}}\abs{\mathcal{G}_j}\hat{s}\paren{\mathcal{G}_j}^d
= \sum_{j=1}^{\Nname{sizes}}\sum_{i=1}^{\Nname{configs}}x_i^*c_{i,j}\hat{s}\paren{\mathcal{G}_j}^{d-k}
= \sum_{i=1}^{\Nname{configs}}\vol[d-k]{P_i} x_i^*\\
\vol{\mathcal L}
&\geq \vol{\mathcal L_\up} - \epsilon \vol{\mathcal L} - \vol[k-1]{B}\hat{s}\\
&= \sum_{i=1}^{\Nname{configs}} \vol[d-k]{P_i} x_i^* - \epsilon \vol{\mathcal L} - \vol[k-1]{B}\hat{s}.
\end{align*}
We can bound the volume $V_1$ of the small items that were packed between the large items using \cref{strip packing:gaps:fvol}:
\begin{align*}
V_1 &:= \sum_{i=1}^{\Nname{configs}} \paren{\vol{L_i} - \fvol{L_i} - \vol[d-k]{P_i} x_i^* }\\
&\geq \sum_{i=1}^{\Nname{configs}} \paren{
x_i^*
- \frac{3}{2}\epsilon x_i^*
- \frac{\epsilon}{d-1} \vol[k-1]{B}\hat{s}
- \vol[d-k]{P_i} x_i^* }
\tag{using $\vol{L_i} = x_i^*$}\\
&= \paren{1 - \frac{3}{2}\epsilon} \paren{\sum_{i=1}^{\Nname{configs}} h_i} \vol[k-1]{B}
- \Nname{configs} \frac{\epsilon}{d-1} \vol[k-1]{B}\hat{s}
- \sum_{i=1}^{\Nname{configs}} \vol[d-k]{P_i} x_i^*
\tag{using $h_i = x_i^*/\vol[k-1]{B}$}\\
&> \paren{1 - \frac{3}{2}\epsilon} \paren{\sum_{i=1}^{\Nname{configs}} h_i} \vol[k-1]{B}
- \Nname{configs} \frac{3}{2}\epsilon \vol[k-1]{B}\hat{s}
- \sum_{i=1}^{\Nname{configs}} \vol[d-k]{P_i} x_i^*.
\end{align*}
The volume $V_2$ of the remaining small and medium items is located in the top layer and we can bound it using \cref{strip packing:top layer:fvol}:
\begin{align*}
V_2 &:= \vol{L_{\mathrm{top}}} - \fvol{L_{\mathrm{top}}}\\
&\geq h_{\mathrm{top}} \vol[k-1]{B}
- \frac{3}{2}\epsilon h_{\mathrm{top}} \vol[k-1]{B}
- \frac{\epsilon}{d-1} \vol[k-1]{B}\hat{s}\\
&= \paren{1 - \frac{3}{2}\epsilon} h_{\mathrm{top}}  \vol[k-1]{B}
-  \frac{\epsilon}{d-1} \vol[k-1]{B}\hat{s}\\
&> \paren{1 - \frac{3}{2}\epsilon} h_{\mathrm{top}}  \vol[k-1]{B}
-  \frac{3}{2}\epsilon \vol[k-1]{B}\hat{s}.
\end{align*}
Recall that $h_\mathrm{total} = \sum_{i=1}^{\Nname{configs}}\paren{h_i + \hat{s}} +  h_\mathrm{top} + \hat{s}$.
The bound on the total volume of small and medium items is therefore:
\begin{align*}
&\quad\vol{\mathcal S} + \vol{\mathcal M}\\
&= V_1 + V_2\\
&> \paren{1 - \frac{3}{2}\epsilon} \paren{h_{\mathrm{total}} - \paren{\Nname{configs} + 1} \hat{s}} \vol[k-1]{B}\\
&\quad - \paren{\Nname{configs} + 1} \frac{3}{2}\epsilon \vol[k-1]{B}\hat{s}
- \sum_{i=1}^{\Nname{configs}} \vol[d-k]{P_i} x_i^*\\
&= \paren{1 - \frac{3}{2}\epsilon} h_{\mathrm{total}} \vol[k-1]{B}
- \paren{\Nname{configs} + 1} \paren{1 - \frac{3}{2}\epsilon} \vol[k-1]{B} \hat{s}\\
&\quad- \paren{\Nname{configs} + 1} \frac{3}{2}\epsilon \vol[k-1]{B}\hat{s}
- \sum_{i=1}^{\Nname{configs}} \vol[d-k]{P_i} x_i^*\\
&= \paren{1 - \frac{3}{2}\epsilon} h_{\mathrm{total}} \vol[k-1]{B}
-\paren{\Nname{configs} + 1} \vol[k-1]{B}\hat{s}
- \sum_{i=1}^{\Nname{configs}} \vol[d-k]{P_i} x_i^*.
\end{align*}
Thus, we have a bound on the volume of all items:
\begin{align*}
\paren{1+\epsilon} \vol{\mathcal I}
&= \vol{\mathcal L} + \vol{\mathcal S} + \vol{\mathcal M} + \epsilon \vol{\mathcal I}\\
&> \sum_{i=1}^{\Nname{configs}} \vol[d-k]{P_i} x_i^*
- \epsilon \vol{\mathcal L}
- \vol[k-1]{B}\hat{s}\\
&\quad+ \paren{1 - \frac{3}{2}\epsilon}
h_{\mathrm{total}} \vol[k-1]{B}
-\paren{\Nname{configs} + 1} \vol[k-1]{B}\hat{s}\\
&\quad- \sum_{i=1}^{\Nname{configs}} \vol[d-k]{P_i} x_i^*
+ \epsilon \vol{\mathcal I}\\
&\geq \paren{1 - \frac{3}{2}\epsilon}
h_{\mathrm{total}} \vol[k-1]{B}
- \paren{\Nname{configs} + 2} \vol[k-1]{B}\hat{s}. 
\end{align*}
This gives us the claimed bound on the height. It holds that
\begin{align*}
h_{\mathrm{total}}
&\leq \frac{1+\epsilon}{1-\frac{3}{2}\epsilon} \frac{\vol{\mathcal I}}{\vol[k-1]{B}}
+ \frac{1}{1-\frac{3}{2}\epsilon}\paren{\Nname{configs} + 2} \hat{s}\\
&\leq \paren{1+5\epsilon}\opts{\mathcal I, A, B}
+ 2\paren{\Nname{configs} + 2} \hat{s}.
\end{align*}
For the last step we use $2\paren{1+\epsilon}/\paren{2-3\epsilon} \leq 1 + 5 \epsilon$ for $\epsilon \leq 1/3$
and the fact that
\begin{align*}
\opts{\mathcal I, A, B}
\geq \frac{\vol{\mathcal I}}{\vol[d-k]{A}\vol[k-1]{B}}
= \frac{\vol{\mathcal I}}{\vol[k-1]{B}}.&\qedhere
\end{align*}
\end{proof}

If all small items were packed in the gaps, the top layer only contains items of medium size.
By choice of $\eta$, we know that \kvn{the medium} items have a small total volume.
Thus we can argue that the height of $h_{\mathrm{top}}$ of the top layer is small, which also gives us a good bound on $h_{\mathrm{total}}$.

\begin{lemma}\label{strip packing:case 2}
If all small items were packed into the gaps of the layers $L_1, \dots, L_{\Nname{configs}}$, then we have a height of
\begin{align*}
h_{\mathrm{total}} \leq \paren{1+\paren{2^{k-1} + 2}\epsilon} \opts{\mathcal I, A, B}
+ \paren{\Nname{configs} + 3}\hat{s}.
\end{align*}
\end{lemma}
\begin{proof}
By the choice of $\mathcal M$ and \cref{strip packing:top layer:fvol} we have
\begin{align*}
\epsilon \vol{\mathcal I} &\geq \vol{\mathcal M}
= \vol{L_{\mathrm{top}}} - \fvol{L_{\mathrm{top}}}\\
&\geq h_{\mathrm{top}}\vol[k-1]{B}
- \frac{3}{2}\epsilon h_{\mathrm{top}}\vol[k-1]{B}
- \frac{\epsilon}{d-1}\vol[k-1]{B}\hat{s}.
\end{align*}
This gives us a bound on the height of the top layer:
\begin{align*}
h_{\mathrm{top}}
&\leq \frac{\epsilon}{1-\frac{3}{2}\epsilon} \frac{\vol{\mathcal I}}{\vol[k-1]{B}}
+ \frac{\epsilon}{1-\frac{3}{2}\epsilon}\frac{1}{d-1}\hat{s}\\
&\leq 2\epsilon \opts{\mathcal I, A, B}
+ \frac{2}{3\paren{d-1}}\hat{s}\\
&< 2\epsilon \opts{\mathcal I, A, B}
+ \hat{s}.
\end{align*}
By applying this, we can prove the claimed bound on the total height:
\begin{align*}
h_{\mathrm{total}}
&= \sum_{i=1}^{\Nname{configs}}h_i
+ h_{\mathrm{top}}
+ \paren{\Nname{configs} + 1}\hat{s}\\
&\leq \opts{\mathcal{L}_\up, A, B}
+ 2\epsilon \opts{\mathcal I, A, B}
+ \paren{\Nname{configs} + 2}\hat{s}
\tag{using \cref{strip packing:large:lp}}\\
&\leq \paren{1+\paren{2+2^{k-1}}\epsilon}\opts{\mathcal I, A, B}
+ \paren{\Nname{configs} + 3}\hat{s}
\tag{using \cref{strip packing:large:rounding} and $\mathcal{L} \subseteq \mathcal{I}$}
\end{align*}
\end{proof}

The only part missing for the proof of \cref{knapsack:structure:strip packing} is a constant bound on $\Nname{configs}$.

\begin{lemma}\label{strip packing:constants}
The numbers $\rho_{\eta}$, $\Nname{sizes}$ and $\Nname{configs}$ can be bounded by constants that are only dependent on $d$, $k$, $\alpha$, $N$ and $\epsilon$.
\end{lemma}
\begin{proof}
In the following we will define new constants formally as functions over the constants $d$, $k$, $\alpha$, $N$ and $\epsilon$ received from the problem definition to show their dependencies.
As those values never change throughout this proof, we will however skip those arguments when using the newly defined constants to gain clearer formulas.
First of all we can bound $\rho_1$ in the following way:
\begin{align*}
\rho_1 &= \frac{\epsilon\hat{s}}{2\left(d-1\right)N\left(\alpha\hat{s}\right)^{d-k}}\\
&= \frac{\epsilon}{2\left(d-1\right)N\alpha^{d-k}\hat{s}^{d-k-1}}\\
&\geq \frac{\epsilon}{2\left(d-1\right)N\alpha^{d-k}}
\tag{using $\hat{s} \leq 1$}\\
&=: \Cname{\rho, 1}\left(d, k, \alpha, N, \epsilon\right).
\end{align*}
Now we can bound $\rho_i \geq \Cname{\rho, i}$ with
\begin{align*}
\Cname{\rho, i}\left(d, k, \alpha, N, \epsilon\right)
:= \frac{\Cname{\rho, 1}^{i\left(d-k\right)^{2i-2}}}
{2^{\left(i-1\right)\left(d-k\right)^{2i-3}}}.
\end{align*}
We prove it by induction. For $i=1$ the new definition of $\Cname{\rho, 1}$ is consistent with the old one and thus the bound is already shown above.
For $i+1$ we can prove the bound with
\begin{align*}
\rho_{i+1} &= \rho_1\frac{\rho_i^{\left(d-k\right)^2}}{2^{d-k}}\\
&\geq \frac{\Cname{\rho, 1}}{2^{d-k}}
\left(\frac{\Cname{\rho, 1}^{i\left(d-k\right)^{2i-2}}}
{2^{\left(i-1\right)\left(d-k\right)^{2i-3}}}\right)^{\left(d-k\right)^2}\\
&= \frac{\Cname{\rho, 1}}{2^{d-k}}
\frac{\Cname{\rho, 1}^{i\left(d-k\right)^{2i}}}
{2^{\left(i-1\right)\left(d-k\right)^{2i-1}}}\\
&\geq\frac{\Cname{\rho, 1}^{\left(i+1\right)\left(d-k\right)^{2i}}}
{2^{i\left(d-k\right)^{2i-1}}}\\
&= \Cname{\rho, i+1}.
\end{align*}
Using $\eta \leq \ceil{1/\epsilon}$ we get
$\rho_{\eta} \geq \Cname{\rho, \ceil{1/\epsilon}}
=: \Cname{\rho}\left(d, k, \alpha, N, \epsilon\right)$.
With this bound we can bound the number of sizes to which we round the large items to
\begin{align*}
\Nname{sizes}
= \ceil{\frac{1}{\epsilon\rho_{\eta}^d}}
\leq \frac{1}{\epsilon \Cname{\rho}^d} + 1
=: \Cname{sizes}\left(d, k, \alpha, N, \epsilon\right).
\end{align*}
The number of configurations can now be bounded by
\begin{align*}
\Nname{configs}
\leq \left(\frac{1}{\rho_{\eta}^{d-k}}+1\right)^{\Nname{sizes}}
\leq \left(\frac{1}{\Cname{\rho}^{d-k}}+1\right)^{\Cname{sizes}}
=: \Cname{configs}\left(d, k, \alpha, N, \epsilon\right).
\end{align*}
\end{proof}

Note, that we have the following property  which is important in \cref{sec:knapsack}.
\begin{lemma}
\label{lem:cconfig}
$\Cname{configs} > 1/\Cname{\rho,1} > N\alpha^{d-k} \geq \alpha$.
\end{lemma}
\begin{proof}
\begin{align*}
\Cname{configs} &=  \left(\frac{1}{\Cname{\rho}^{d-k}}+1\right)^{\Cname{sizes}} \tag{From definition of $\Cname{configs}$}\\
&\ge {1}/{\Cname{\rho}} \tag{As $\Cname{\rho} \le 1, \Cname{sizes} \ge 1$}\\
& \ge 1/\Cname{\rho,1} \tag{As $\Cname{\rho,1} \ge \Cname{\rho}$} \\
&= 2(d-1)N \alpha^{d-k}/\eps \tag{From definition of $\Cname{\rho,1}$}\\
&> N\alpha^{d-k} \tag{As $d \ge 2, \eps<1$} \\
&\geq \alpha \tag{As $N \ge 1, \alpha \ge 1$}
\end{align*}
\end{proof}
\arir{added this proof. Please check. \mar{checked.}}

\begin{lemma}\label{lem:knapsack:structure:strip packing:boxes}
The strip packing algorithm of \cref{sec:strip packing} packs all items in \Nboxes{} and \Vboxes{}.
The number of \Nboxes{} is bounded by
$
\Cname{configs}/C_{\rho}^{d-k}
$
and the number of \Vboxes{} is bounded by
$
\Cname{configs}N2^{d-k} / \Cname{\rho}^{\left(d-k\right)^2} + N
$.
\end{lemma}
\begin{proof}
Note that the strip packing algorithm creates \Nboxes{} only at the end of \cref{sec:strip packing:large}.
In this step, one \Nbox{} is created for each large item in each configuration.
As mentioned before, the volume of $A$ is $1$ and each large item has a side length of at least $\rho_\eta$.
Thus, there are at most $1/\Cname{\rho}^{d-k}$ large items per configuration and there are at most $\Cname{configs}$ configurations.
This results in an upper bound of $\Cname{configs}/C_{\rho}^{d-k}$ for the number of \Nboxes{}.

\Vboxes{} are created in two steps of the algorithm.
The first step is stated in \cref{sec:strip packing:gaps} where we create a \Vbox{} for each gap in each layer that was created for the large items.
There is one layer for each configuration and in the beginning of \cref{sec:strip packing:gaps}, the number of gaps per layer is bounded by $2^{d-k}/\rho_{\eta}^{\paren{d-k}^2}$.
Thus, there are at most $\Cname{configs}N\left(2^{d-k} / \Cname{\rho}^{\left(d-k\right)^2}\right)$ \Vboxes{} created in this step.
The second step where \Vboxes{} are created is in \cref{sec:strip packing:top} where we place the remaining small and medium items in an additional layer.
There, the number of \Vboxes{} is bounded by the number of cells in $A$ which is bounded by $N$.
This concludes our proof.
\end{proof}

\end{document}